\documentclass[11pt]{article}
\usepackage[all]{xy}
\usepackage{amsmath,amsthm,amsfonts,amssymb,amsopn,amscd}     

\DeclareMathOperator{\Tr}{Tr}
\DeclareMathOperator{\Str}{Tr_{s}}

\DeclareMathOperator{\Dim}{dim}
\DeclareMathOperator{\ind}{ind}
\DeclareMathOperator{\fn}{null}

\def\qed{{\unskip\nobreak\hfil\penalty50
\hskip2em\hbox{}\nobreak\hfil$\square$
\parfillskip=0pt \finalhyphendemerits=0\par}\medskip}
\def\proof{\trivlist \item[\hskip \labelsep{\bf Proof.\ }]}
\def\endproof{\null\hfill\qed\endtrivlist\noindent}

\def\Ad{{\hbox{\rm Ad}}}

\def\dim{{\hbox{dim}}}

\def\Hom{{\hbox{Hom}}}

\def\a{\alpha}

\def\de{\delta}
\def\e{\varepsilon}
\def\g{\gamma}
\def\Ga{\Gamma}
\def\k{\kappa}
\def\l{\lambda}

\def\r{\rho}
\def\phi{\varphi}

\def\th{\theta}

\def\Om{\Omega}

\newtheorem{theorem}{Theorem}
\newtheorem{lemma}[theorem]{Lemma}
\newtheorem{corollary}[theorem]{Corollary}

\newtheorem{proposition}[theorem]{Proposition}
\newtheorem{sublemma}[theorem]{Sublemma}

\def\setminus{\smallsetminus}
\def\A{{\cal A}}
\def\B{{\cal B}}
\def\C{{\cal C}}

\def\L{{\cal L}}
\def\I{{\cal I}}
\def\H{{\cal H}}
\def\K{{\cal K}}
\def\S{{\cal S}}
\def\U{{\cal U}}

\def\s{{\sigma}}

\def\emptyset{\varnothing}

\def\Diff{{\rm Diff}}

\def\Mob{{\rm\textsf{M\"ob}}}
\def\Mobn{{\rm\textsf{M\"ob}}^{(n)}}
\def\Mobb{{\rm\textsf{M\"ob}}^{(2)}}
\def\Mobi{{\rm\textsf{M\"ob}}^{(\infty)}}
\def\Si{S^{1(\infty)}}
\def\Sn{S^{1(n)}}
\def\S2{S^{1(2)}}

\def\emptyset{\varnothing}
\def\setminus{\smallsetminus}

\def\Diff{{\mathrm {Diff}}}

\def\SVir{{\mathrm {SVir}}}
\def\Z{{\mathbb Z}}
\def\la{\lambda}

\def\RR{{\mathbb R}}

\def\ZZ{{\mathbb Z}}

\def\sl2{{{\rm SL}(2,\RR)}}
\def\psl2{{{\rm PSL}(2,\RR)}}

\def\u1{{{\rm U}(1)}}
\def\su2{{{\rm SU}(2)}}

\def\so3{{{\rm SO}(3)}}

\def\A{{\mathcal A}}
\def\B{{\mathcal B}}
\def\C{{\mathcal C}}

\def\F{{\mathcal F}}
\def\H{{\mathcal H}}
\def\I{{\mathcal I}}
\def\K{{\mathcal K}}
\def\k{K}

\def\U{{\mathcal U}}

\title{{\Huge Structure and Classification of Superconformal Nets}}
\author{
{\sc Sebastiano Carpi}\footnote{Supported by MIUR, GNAMPA-INDAM and EU network ``Noncommutative Geometry" 
MRTN-CT-2006-0031962}\\
Dipartimento di Scienze\\
Universit\`a di Chieti-Pescara ``G. d'Annunzio''\\
Viale Pindaro, 42, I-65127 Pescara, Italy\\
E-mail: {\tt carpi@sci.unich.it}\\
\phantom{X}\\
{\sc Yasuyuki Kawahigashi}\footnote{Supported in part by the
Grants-in-Aid for Scientific Research, JSPS.}\\
Department of Mathematical Sciences\\
University of Tokyo, Komaba, Tokyo, 153-8914, Japan\\
E-mail: {\tt yasuyuki@ms.u-tokyo.ac.jp}\\
\phantom{X}\\
{\sc Roberto Longo}$^*$
\\
Dipartimento di Matematica,
Universit\`a di Roma ``Tor Vergata'',\\
Via della Ricerca Scientifica, 1, I-00133 Roma, Italy\\
E-mail: {\tt longo@mat.uniroma2.it}}

\begin{document}
\date{}
\maketitle

\begin{abstract}
We study the general structure of Fermi conformal nets of von Neumann algebras on $S^1$, consider a class of topological representations, the general representations, that we characterize as Neveu-Schwarz or Ramond representations, in particular a Jones index can be associated with each of them. We then consider a supersymmetric general representation associated with a Fermi modular net and give a formula involving the Fredholm index of the supercharge operator and the Jones index. We then consider the net associated with the super-Virasoro algebra and discuss its structure. If the central charge $c$ belongs to the discrete series, this net is modular by the work of F. Xu and we get an example where our setting is verified by considering the Ramond irreducible representation with lowest weight $c/24$. We classify all the irreducible Fermi extensions of any super-Virasoro net in the discrete series, thus providing a classification of all superconformal nets with central charge less than $3/2$.
\end{abstract}
\newpage
{\small\tableofcontents}
\section{Introduction}
As is well known the r\^ole of symmetry in Physics is fundamental. Symmetries are usually divided in two types: spacetime and internal symmetries. 

If we ignore the effect of gravity that provides spacetime curvature, Special Relativity tells us that our spacetime is the flat Minkowski space and its symmetries are given by the Poincar\'e group. However, if no massive particle is present, the symmetry group can be larger. The largest transformation group compatible with locality is the conformal group. The conformal group is a finite-dimensional Lie group; its symmetry action is so stringent that one expects only few conformal Quantum Field Theory models to exist in four dimensions, see \cite{BRNT} and references therein.

A much richer structure comes out in lower dimensional spacetimes. On the two-dimensional Minkowski spacetime the conformal group is infinite dimensional and factors as a product of the diffeomorphism groups of the light-like (chiral) lines $x\pm t=0$ and acts on the their compactification, namely one gets an action of $\Diff(S^1)\times\Diff(S^1)$. Restricting a quantum field to a chiral line gives rise to a quantum field theory on $S^1$. 

We do not dwell here on crucial r\^ole played by Conformal Quantum Field Theory in various physical contexts, nor on its deep mathematical connections with other subjects.
We mention however a less emphasised r\^ole of the subject as a laboratory for more general analysis. This is due to the variety of constructed models and the powerful methods of analysis. We are interested here in particular in Operator Algebraic methods whose effectiveness has  also been shown by recent classification results and new model constructions \cite{KL1,KLPR,X6}.

Internal symmetries are also fundamental. In particular, in Quantum Field Theory, they are basically expressed by a gauge group. Concerning the particle-antiparticle symmetry, it showed up with the discovery of the Dirac equation and is, in a sense, a derived symmetry. There are general arguments to link the particle-antiparticle symmetry with the spacetime symmetries, see \cite{GL1}.

Now there is an important higher level symmetry relating spacetime and internal symmetries: supersymmetry, see \cite{WB}. Supersymmetry sets up a correspondence between Bose and Fermi particles. Leaving aside important physical aspects of supersymmetry, e.g. in relation to the Higgs particle, we wish to mention that the related mathematical structure has profound noncommutative geometrical implications, see \cite{C,JLO}.

If one considers a quantum field theory which is both chiral conformal and supersymmetric, one then expects a very stringent interesting structure to show up. Superconformal models have indeed been studied by different approaches. For example, the tricritical Ising model is a basic model with a remarkable structure \cite{FQS}.

The purpose of this paper is to initiate a general, model independent, operator algebraic study of superconformal quantum field theory on the circle and pursue this analysis to the classification of the superconformal nets in the discrete series. 

In the first sections we describe the general property of a Fermi conformal net on $S^1$. Basic properties, already considered in the non-local case \cite{DLR,LR2}, are here described in our context (Bisognano-Wichmann property, twisted Haag duality,\dots).

A Fermi net $\A$ contains a local subnet $\A_b$, the Bose subnet (the fixed-point under the grading symmetry) which is automatically equipped with an involutive sector $\s$, dual to the grading. This splits the sectors of $\A_b$ in two subsets, $\s$-Bose and $\s$-Fermi sectors, that can be studied in the standard way as $\A_b$ is local.

For a general model analysis one would like to consider representations of $\A$. There is an obvious extension of the notion of DHR representation to Fermi nets, that we study at the beginning. However, as we shall see, there is more general natural class of representations for Fermi net: the \emph{general representations}. These are representations of the promotion $\A^{(n)}$ of $\A$ to the $n$-cover $S^{1(n)}$ of $S^1$, that restrict to DHR representations of the Bose subnet $\A_b$. We shall see that a general representation is indeed  a representation of the double cover net $\A^{(2)}$ and can be equivalently described as a \emph{general soliton}. This splits the general representations in two classes: Neveu-Schwarz (DHR) and Ramond (properly associated with $\A^{(2)}$). We shall later see that a representation is Neveu-Schwarz (resp. Ramond) iff it comes by $\a$-extension by a $\s$-Bose (resp. $\s$-Fermi) representation of $\A_b$.

Having clarified the structure of the representations of a Fermi net, we consider the case where supersymmetry is present. In particular we consider a \emph{supersymmetric representation} of a conformal Fermi net $\A$, namely a graded general representation $\l$ of $\A$ where 
\[
\tilde H_\l \equiv H_\l - \frac{c}{24} = Q_\l^2 \ ;
\]
here $H_\l$ is the conformal Hamiltonian and $Q$ is an odd selfadjoint operator (the supercharge). In this case the McKean-Singer lemma (Appendix \ref{MKS}) gives
\[
\Str(e^{-t\tilde H_\l}) = \text{Fredholm index of}\ Q_{\l +}\ , 
\]
for all $t>0$, where $\Str$ denotes the supertrace. In particular the left hand side, also called the Witten index, is an integer, the Fredholm index $\ind(Q_{\l +})$ of upper off diagonal part of $Q_\l$.

If $\l$ is irreducible, its restriction $\l_b$ to $\A_b$ has two irreducible components $\l_b = \r\oplus\r'$ and we can consider the Jones index of $\r$, which is the square of the Doplicher-Haag-Roberts statistical dimension $d(\r)$ \cite{L5}.

If $\A_b$ is modular, we can express $d(\r)$ by the Kac-Peterson Verlinde matrix $S$ as this is equal to the Rehren matrix $S$. By combining all this information and an argument in \cite{KL05}, we then get the formula
\[
\ind(Q_{\l +}) = 
\frac{d(\r)}{\sqrt{\mu_\A}}\sum_{\nu\in\mathfrak R}
K(\r,\nu)d(\nu) 
\]
where $K$ is a matrix that is related to $S$, $\mu_\A$ is the $\mu$-index \cite{KLM} and the sum is over all Ramond irreducible sectors. This formula thus involves both the Fredholm index and the Jones index.

We then put our attention towards model analysis and aiming firstly to illustrate concrete situations where our general setting is realised.

The infinite-dimensional super-Lie algebra that governs the superconformal symmetries is the super-Virasoro algebra. It is a central extension by a central element $c$ of the Lie algebra generated by the Fourier modes $L_n$ and $G_r$ of the Bose and the Fermi stress-energy tensor. Clearly the $L_n$ generate the Virasoro algebra, thus the super-Virasoro algebra contains the Virasoro algebra; the commutation relation are given by the equations \eqref{svirdef}. It has been studied in \cite{GKO, FQS}.

The super-Virasoro algebra plays for superconformal fields the universal same r\^ole that the Virasoro algebra plays for local conformal fields. Our initial task is then to construct and analyse the super-Virasoro nets. Following the work \cite{GKO} we describe the nets; if the central charge $c$ is less than $3/2$ we identify the Bose subnet as a coset, study the representation structure and show the modularity of the net. There is a distinguished representation with lowest weight $c/24$ which is supersymmetric and thus provides an interesting example for our index formula. 

Now, if $c<3/2$ we are in the discrete series \cite{FQS}; analogously to the local conformal case \cite{KL1} every superconformal net is an irreducible finite-index extension of a super-Virasoro net. We classify all such extensions. There are two series of extensions, the one given by the trivial extension and a second one given by index 2 extensions. Beside these there are six exceptional extensions that we describe explicitly.
\section{Fermi nets on $S^1$}
In this section we discuss the basic notions for a Fermi conformal net of von Neumann algebras on $S^1\equiv \{z\in\mathbb C: |z|=1\}$, and their first implications. It is convenient to start by analysing M\"obius covariant nets. Note that a general discussion of M\"obius covariant net is contained in \cite{DLR}, here however we focus on the Fermi case.
\subsection{M\"obius covariant Fermi nets}
\label{MobFermiNets}
We shall denote by $\Mob$ the M\"obius group, which is isomorphic to $SL(2,\mathbb R)/\mathbb Z_2$ and acts naturally and faithfully on the circle $S^1$. The $n$-cover of $\Mob$ is denoted by $\Mobn$, $n\in\mathbb N\cup\{\infty\}$. Thus $\Mob^{(2)}\simeq SL(2,\mathbb R)$ and $\Mobi$ is the universal cover of $\Mob$.

By an interval of $S^1$ we mean, as usual, a non-empty, non-dense, open, connected subset of $S^1$ and we denote by $\I$ the set of all intervals. If $I\in\I$, then also $I'\in\I$ where $I'$ is the interior of the complement of $I$. 

A {\it  net $\A$ of von Neumann algebras on $S^1$} is a map
\[
I\in\I\mapsto\A(I)
\]
from the set of intervals to the set of von Neumann algebras on a
(fixed) Hilbert space $\H$ which verifies the \emph{isotony property}:
\[
I_1\subset I_2\Rightarrow  \A(I_1)\subset\A(I_2)
\]
where $I_1 , I_2\in\I$.

A \emph{M\"obius covariant net} $\A$ of von Neumann algebras on $S^1$ is a net of von Neumann algebras on $S^1$ such that the following properties $1-4$ hold:
\begin{description}
\item[\textnormal{\textsc{1. M\"obius covariance}}:] {\it There is a
strongly continuous unitary representation $U$ of} $\Mobi$ {\it on $\H$ such that}
\[
U(g)\A(I)U(g)^*=\A(\dot{g}I)\ ,
\qquad g\in \Mobi,\ I\in\I \ ,
\]
{\it where $\dot{g}$ is the image of $g$ in} $\Mob$ {\it under the quotient map.}
\end{description}
\begin{description}
\item[$\textnormal{\textsc{2. Positivity of the energy}}:$]
{\it The generator of the rotation one-para\-meter subgroup 
$\theta\mapsto U({\rm rot}^{(\infty)}(\theta))$ 
(conformal Hamiltonian) is positive}, na\-me\-ly $U$ is a positive energy representation
\footnote{Here ${\rm rot}^{(n)}(\theta)$ is the lift to $\Mobn$ of the $\theta$-rotation ${\rm rot}(\theta)$ of $\Mob$. For shortness we sometime write ${\rm rot}^{(n)}(\theta)$ simply by ${\rm rot}(\theta)$ and $U(\theta) = U({\rm rot}(\theta))$.}. 
\end{description}
\begin{description}
\item[\textnormal{\textsc{3. Existence and uniqueness of the vacuum}}:]
{\it There exists a unit $U$-invariant vector $\Omega$ 
(vacuum vector), unique up to a phase, and $\Omega$ is
cyclic for the von~Neumann algebra $\vee_{I\in\I}\A(I)$}
\end{description}
Given a M\"obius covariant net $\A$ in $S^1$, a unitary $\Ga$ such that $\Ga\Om = \Om$, $\Ga\A(I)\Ga^* =\A(I)$ for all $I\in\I$ is called a gauge unitary and the adjoint net automorphism $\g\equiv\Ad\Ga$ a \emph{gauge} automorphism\footnote{The requirement $\Ga\Om =\Om$ (up to a phase that can be put equal to one) is equivalent to $\Ga U(g) = U(g)\Ga$ by the uniqueness of the representation $U$ in Cor. \ref{unique}; we shall use the second formulation in the non-vacuum case.}. 

A $\mathbb Z_2$-grading on $\A$ is an involutive gauge automorphism $\g$ of $\A$. Given the grading $\g$, an element $x$ of $\A$ such that $\g(x)=\pm x$ is called homogeneous, indeed a Bose or Fermi element according to the $\pm$ alternative. 
We shall say that the degree $\partial x$ of the homogeneous element $x$ is $0$ in the Bose case and $1$ in the Fermi case.

Every element $x$ of $\A$ is uniquely the sum $x = x_0 + x_1$ with $\partial x_k=k$, indeed 
$x_k = \big(x  + (-1)^k \g(x)\big)/2$.

A  {\em M\"obius covariant Fermi net} $\A$ on $S^1$ is a $\mathbb Z_2$-graded M\"obius covariant net satisfying graded locality, namely a  M\"obius covariant net of 
von Neumann algebras on $S^1$ such that the following holds:
\begin{description}
\item[\textnormal{\textsc{4. Graded locality}}:]
{\it There exists a grading automorphism $\g$ of $\A$ such that,  if $I_1$ and $I_2$ are disjoint intervals,}
\[
[x,y] = 0,\quad x\in\A(I_1), y\in\A(I_2) \ .
\]
\end{description}
Here $[x,y]$ is the graded commutator with respect to the grading automorphism $\g$ defined as follows: if $x,y$ are homogeneous then
\[
[x,y]\equiv xy - (-1)^{\partial x \cdot \partial y}yx
\]
and, for the general elements $x,y$, is  extended by linearity.

Note the \emph{Bose subnet} $\A_b$, namely the $\g$-fixed point subnet $\A^\g$ of degree zero elements, is local. Moreover, setting
\[
Z\equiv \frac{1 - i\Ga}{1 - i}
\]
we have that the unitary $Z$ fixes $\Omega$ and
\[
\A(I')\subset Z\A(I)'Z^*,
\]
(twisted localiy w.r.t. $Z$), that is indeed equivalent to graded locality.
\medskip

\noindent
{\it Remark.} Strictly speaking a M\"obius covariant net is a pair $(\A, U)$ where $\A$ is a net and $U$ is a unitary representation of $\Mob$ that satisfy the above properties. In most cases it is convenient to denote the M\"obius covariant net simply by $\A$ and then consider the unitary representation $U$ (note that $U$ is unique once we fix the vacuum vector). However, later in this paper, it will be convenient to use the notation $(\A, U)$ (also in the diffeomorphism covariant case).
\subsection{First consequences}
We collect here a few first consequences of the axioms of a M\"obius covariant Fermi net on $S^1$. They can be mostly derived by a simple extension of the proofs in the local case, cf. \cite{DLR,LR2}.
\subsubsection*{Reeh-Schlieder theorem}
\begin{theorem}\label{Reeh-Schlieder} Let $\A$ be a M\"obius covariant Fermi net on
$S^1$.  Then $\Omega$ is cyclic and separating for each von Neumann
algebra $\A(I)$, $I\in\I$.
 \end{theorem}
 \proof
The cyclicity of $\Omega$ follows exactly as in the local case. By twisted locality, then $\Omega$ is separating too.
\endproof
\subsubsection*{Bisognano-Wichmann property}
If $I\in \I$, we shall denote by $\Lambda_I$ the one parameter subgroup of $\Mob$ of ``dilation associated with $I$\ \!'', see \cite{BGL}. We shall use the same symbol also to denote the unique one parameter subgroup of $\Mobn$  ($n$ finite or infinite) that project onto $\Lambda_I$ under the quotient map. The following theorem and its corollary are proved as usual, see \cite{DLR}.
\begin{theorem}
Let $I\in\I$ and $\Delta_I$, $J_I$ be the modular operator and the modular 
conjugation of $(\A(I),\Omega)$. Then we have:

$(i)$: 
\begin{equation}\label{BW}
\Delta_{I}^{it} = U(\Lambda_I(-2\pi t)), \ t\in\mathbb R,
\end{equation}

$(ii)$: 
$U$ extends to an (anti-)unitary representation of {\rm $\Mobi\ltimes\mathbb Z_2$}
determined by
\[
U(r_I)=ZJ_I,\ I\in\I,
\]
acting covariantly on $\A$, namely
\[
U(g)\A(I)U(g)^*=\A(\dot{g}I)\quad g\in\text{\rm $\Mobi$}\ltimes\mathbb Z_2\ I\in \I\ .
\]
Here $r_I:S^1\to S^1$ is the reflection mapping $I$ onto $I'$, see \cite{BGL}.
\end{theorem}
\begin{corollary}\label{unique} {\em (Uniqueness of the M\"obius representation)}
The representation $U$ is unique.
\end{corollary}
\begin{corollary} {\em (Additivity)}
  Let $I$ and $I_i$ be intervals with $I\subset\cup_i I_i$. Then 
  $\A(I)\subset\vee_i\A(I_i)$.
\end{corollary}
\medskip

\noindent
{\it Remark.} If $E\subset S^1$ is any set, we denote by $\A(E)$ the von Neumann algebra generated by the $\A(I)$'s as $I$ varies in the intervals $I\in\I$, $I\subset E$. It follows easily by M\"obius covariance that if $I_0\in\I$ has closure $\bar I_0$, then $\A(\bar I_0)=\bigcap_{I\supset {\bar I_0}, I\in\I}\A(I)$.
\subsubsection*{Twisted  Haag duality}
\begin{theorem}
For every  $I\in\I$, we have:
\[
\A(I')=Z\A(I)'Z^*
\]
\end{theorem}
\proof As already noted twisted locality $\A(I)' \supset Z\A(I')Z^*$ holds as a consequence of the Bose-Fermi commutation relations. By the Bisognano-Wichmann property, $Z\A(I')Z^*$ is a von Neumann subalgebra of $\A(I)'$ globally invariant under the vacuum modular group. As the vacuum vector is cyclic for $\A(I')$ by the Reeh-Schlieder theorem, we have $\A(I)' = Z\A(I')Z^*$ by the Tomita-Takesaki modular theory.
\endproof
In the following corollary, the grading and the graded commutator is considered on $B(\H)$ w.r.t. $\Ad\Ga$.
\begin{corollary}\label{td}
$\A(I')= \big\{x\in B(\H):\ [x,y]=0\ \forall y\in\A(I)\big\}$.
\end{corollary}
\proof
Set $X\equiv\{x\in B(\H):\ [x,y]=0\ \forall y\in\A(I)\}$. We have to show that $X\subset\A(I')$. 
If $x = x = x_0 + x_1\in X$ then $Z^*xZ =  x_0 - ix_1\Ga$. As $[x,y]=0$ for all $y\in\A(I)$ it follows that $Z^*xZ\in\A(I)'$ so $x\in Z\A(I)'Z^* = \A(I')$.
\endproof
\subsubsection*{Irreducibility}
We shall say that $\A$ is \emph{irreducible} if the von Neumann algebra $\vee\A(I)$
generated by all local algebras coincides with $B(\H)$.

The irreducibility property is indeed equivalent to several other requirements.
\begin{proposition}\label{irr} Assume all properties $1-4$ for $\A$ except the uniqueness of the vacuum. The following are equivalent:
\begin{itemize}     
\item[{$(i)$}] $\mathbb C\Omega$ are the only $U$-invariant vectors.
\item[{$(ii)$}] The algebras $\A(I)$, $I\in\I$, are
factors. In particular they are type III$_1$ factors (or ${\rm dim}\, \H = 1$).
\item[{$(iii)$}] The net $\A$ is irreducible.
\end{itemize}
\end{proposition}
\proof See \cite{DLR}. 
\endproof
\subsubsection*{Vacuum spin-statistics relation}
\label{vacss} 
The relation $U(2\pi)=1$ holds in the local case \cite{GL2}; in the graded local case it generalizes to
\[
U(4\pi)=1\ ,
\]
where $U(s) \equiv U({\rm rot}(s)) = e^{isL_0}$ is the rotation one-parameter unitary group, see \cite{DLR}. Indeed we have the following.
\begin{proposition}\label{vss}
Let $\A$ be M\"obius covariant Fermi net. Then:
\[
U(2\pi) =\Ga\ .
\]
\end{proposition}
\proof
By restricting to $\A_b$ the identity representation of $\A$ we get the direct sum $\iota\oplus\s$ of the identity and an automorphism. Clearly $\s$ has Fermi statistics because $\A$ has Bose-Fermi commutation relations. Thus $U = U_\iota \oplus U_\s$; by the conformal spin-statistics theorem \cite{GL2} we then have $U(2\pi)= U_\iota (2\pi) \oplus U_\s (2\pi) = 1\oplus -1 = \Ga$.
\endproof
\begin{corollary} {\em (Uniqueness of the grading)}
The grading automorphism is unique.
\end{corollary}
\proof
Immediate from the uniqueness of the M\"obius unitary representation and the spin-statistics relation $U(2\pi) =\Ga$.
\endproof
\subsection{Fermi conformal nets on $S^1$}
If $\A$ is a M\"obius covariant Fermi net on $S^1$ then, by the spin-statistics relation, the unitary representation of $\Mobi$ is indeed a representation of $\Mobb$. As $\Mob$ is naturally a subgroup of the group $\Diff(S^1)$ of orientation preserving diffeomorphisms of $S^1$, clearly $\Mobb$  is naturally a subgroup of $\Diff^{(2)}(S^1)$, the 2-cover of $\Diff(S^1)$.

Given an interval $I\in\I$, we shall denote by $\Diff_I(S^1)$ the subgroup of diffeomorphisms $g$ of $S^1$ localised in $I$, namely $g(t)=t$ for all $t\in I'$, and by
$\Diff_I^{(2)}(S^1)$ the connected component of the identity of the pre-image of 
$\Diff_I(S^1)$ in $\Diff^{(2)}(S^1)$.

A \emph{Fermi conformal net} $\A$ (of von Neumann algebras) on $S^1$ is a M\"obius covariant Fermi net of von Neumann algebras on $S^1$ such that the following holds:
\begin{description}
\item[\textnormal{\textsc{5. Diffeomorphism covariance}}:]
{\it There exists a projective unitary representation $U$ of $\Diff^{(2)}(S^1)$ on $\H$, extending the unitary representation of} $\Mobb$, {\it such that
\[
U(g)\A(I)U(g)^* = \A(\dot{g}I),\ g\in\Diff^{(2)}(S^1),\ I\in\I,
\]
and
\[
U(g)xU(g)^* = x, \ x\in\A(I'),\ g\in\Diff_I^{(2)}(S^1), \ I\in\I\ .
\]}
\end{description}
Here $\dot{g}$ denotes the image of $g$ in $\Diff(S^1)$ under the quotient map.
\begin{lemma}
For every $g\in\Diff^{(2)}(S^1)$ we have $U(g)\Ga = \Ga U(g)$. In particular,
if $g\in\Diff_I^{(2)}(S^1)$, then $U(g)\in\A_b(I)$.
\end{lemma}
\proof
As $\Ga= U(2\pi)$, we have
$
\Ga U(g) \Ga^* = U(2\pi)U(g)U(2\pi)^*=\chi(g) U(g)
$ for all  $g\in\Diff^{(2)}(S^1)$,
where $\chi$ is a continuos scalar function on $\Diff^{(2)}(S^1)$. As $\Ga^2 =1$, we have $\chi(g)^2 =1$ for all $g$, thus $\chi(g) =1$ because $\Diff^{(2)}(S^1)$ is connected and $\chi(g) =1$ is $g$ is the identity.

With $g\in\Diff_I^{(2)}(S^1)$, then $U(g)$ commutes with $\Ga$, hence with $Z$. By the covariance condition $U(g)$ commutes with $\A(I')$, hence with $Z\A(I')Z^*$, so $U(g)\in\A(I)$ by twisted Haag duality.
\endproof
By the above lemma, the representation of the diffeomorphism group belongs to the Bose subnet, so we may apply the uniqueness result in the local case in \cite{CW} and get the following:
\begin{corollary} \emph{(Uniqueness of the diffeomorphism representation)}
The projective unitary representation $U$ of $\Diff^{(2)}(S^1)$ is unique (up to a projective phase).
\end{corollary}
\proof
With $E = (1 +\Ga)/2$ the orthogonal projection onto $\overline{\A_b \Omega}$, clearly $U|_{E\H}$ is the projective unitary representation of $\Diff(S^1)$ associated with $\A_b$, unique by \cite{CW,W2}. As $\Omega$ is separating for the local algebras, $U(g)E$ determines $U(g)$ if 
$g\in\Diff_I^{(2)}(S^1)$ for any interval $I\in\I$. By Prop. \ref{locdiff}, as $I$ varies, $\Diff_I(S^1)$ algebraically generates all $\Diff(S^1)$, and $U$ is so determined up to a phase.
\endproof
\subsection{DHR representations}
There is a natural notion of representation for a Fermi net which is the 
straightforward extension of the notion of representation for a local net, see \eqref{rep} here below. We shall see in later sections that is important to consider also more general representations when dealing with a Fermi net. In this section, however, we deal with the most obvious notion.

In the following we assume that $\A$ be a Fermi conformal net, although certain notions and results are obviously valid in the M\"obius covariant case.

A \emph{representation} $\l$ of $\A$ is a map $I\to\l_I$ that 
associates to an interval $I$ of $S^1$ a normal\footnote{The normality of $\l_I$ is automatic if $\H_\l$ is separable because $\A(I)$ is a type $III$ factor.}
representation $\l_I$ of $\A(I)$ on a fixed Hilbert space $\H_{\l}$ such that
\begin{equation}
\label{rep}
\l_{\tilde I}|_{\A(I)} = \l_I,\quad I\subset \tilde I\ .
\end{equation}
We shall say that a representation $\l$ on $\H_{\l}$ 
is {\em diffeomorphism covariant} 
if there exists a projective unitary representation $U_\l$ of the universal cover 
$\Diff^{(\infty)}(S^1)$ of $\Diff(S^1)$ on $\H_{\l}$ such that
\[
\l_{\dot{g}I}\big(U(g)xU(g)^*\big) = U_\l(g)\l_I(x)U_\l(g)^*\ ,\  x\in\A(I) , 
\forall g\in \Diff^{(\infty)}(S^1)\ .
\]
Here $\dot g$ denotes the image of $g$ in $\Diff(S^1)$ under the quotient map.
A M\"obius covariant representation is analogously defined.

As we shall later deal with more general representations, we may sometime emphasise that we are considering a representation $\l$ as above, by saying that $\l$ is a DHR representation, the `DHR' being however pleonastic.

We shall say that a representation $\l$ on $\H_{\l}$ is \emph{graded}
if there exists a unitary $\Ga_{\l}$ on $\H_{\l}$ such that
\[
\l_I(\g(x)) = \Ga_{\l}\l_I(x)\Ga^*_{\l},\quad x\in\A(I) ,
\]
for all $I\in\I$. As $\g$ is involutive one may then also choose a selfadjoint $\Ga_\l$.
\begin{proposition}\label{DHRgraded}
Let $\l$ be an irreducible DHR representation of the Fermi conformal net $\A$. Then $\l$ is graded and diffeomorphism covariant with positive energy. Moreover we may take $\Ga_\l = U_\l(2\pi)$.
\end{proposition}
\proof
We first use an argument in \cite{KL1}. Let $I\in\I$ and $g\in\Diff_I^{(2)}(S^1)$. For every $x\in\A(\tilde I)$ with $\tilde I\supset I$ we have $\l_I(U(g))\l_{\tilde I}(x)\l_I(U(g))^* = \l_{\tilde I}(U(g)xU(g))^* = \l_{gI}(U(g)xU(g))^*$. As the group generated by $\Diff_I^{(2)}(S^1)$ as $I$ varies in $\I$ is the entire $\Diff^{(2)}(S^1)$, we see every diffeomorphism is implemented in the representation $\l$ by a unitary $U_\l(g)$, which is equal to $\l_I(U(g))$ if $g$ is localised in $I$.

It remains to show that, by multiplying $U_\l(g)$ by a phase factor, we can get a continuous projective representation with positive energy.
In the local case, the automatic diffeomorphism covariance is proved in \cite{DFK} and the automatic positivity of the energy in \cite{W}. Let $U_{\l_b}$ be the covariance unitary representation associated with $\l_b\equiv \l |_{\A_b}$. If $g$ is localised in $I$ then $U_\l(g)U_{\l_b}(g)^*\in\l_I\big(\A_b(I)'\cap\A(I)\big)=\mathbb C$ (cf. Lemma \ref{outer}), so we are done by replacing $U_\l$ with $U_{\l_b}$.

Finally notice that, by diffeomorphism covariance, it follows that 
\[
U_\l(2\pi)\l_I(x)U_\l(2\pi)^* = \l_I(\g(x))
\]
due to Prop. \ref{vss}. So we may take $\Ga_\l = U_\l(2\pi)$.
\endproof
A \emph{localised endomorphism} of $\A$ is an endomorphism $\r$ of the universal $C^*$-algebra $C^*(\A)$ such that $\r|_{\A(I')}$ is the identity for some $I\in\I$ (then we say that $\r$ is localised in the interval $I$). In other words, $\r$ is a representation such that $\r_{I'}= \iota $ and $\r_{\tilde I}$ maps $\A(\tilde I)$ into itself if $\tilde I\supset I$. This last property is automatic by Haag duality in the local case, and we now see to hold also in the Fermi case.

Next lemma shows that the grading is locally outer. This applies indeed to every gauge automorphism, 
see \cite{Car99,X01}.
\begin{lemma}\label{outer}
Given any interval $I$, $\g|_{\A(I)}$ is an outer automorphism (unless the grading is trivial). As a consequence $\A_b(I)'\cap\A(I)=\mathbb C$.
\end{lemma}
\proof
Suppose $\g|_{\A(I)}$ is inner; then there exist unitaries $u\in\A(I)$ and $u'\in\A(I)'$ such that $\Ga =  u'u$. In fact $u$ and $u'$ are unique up to a phase factor because $\A(I)$ is a factor. Now $\Ga$ commutes with the M\"obius unitary group, so $\Delta_I^{it}u\Omega = e^{iat}u\Omega$ for some $a\in\mathbb R$. But $\log\Delta_I$ has no non-zero eigenvalue (see \cite{DLR}), so $u\Omega \in\mathbb C\Omega$, thus $u$ is a scalar $\g$ is the identity.
\endproof
The following proposition generalizes to the Fermi case the well known DHR argument for the correspondence between representations and localised endomorphisms on the Minkowski space. 
\begin{proposition}\label{dhrendo}
Let $\pi$ be a representation of the Fermi conformal net $\A$ on $S^1$, and suppose the Hilbert spaces $\H$ and $\H_\pi$ to be separable. Given an interval $I$, there exists a localised endomorphism of $\A$ unitarily equivalent to $\pi$.
\end{proposition}
\proof
Let $\Ga$ be the unitary, $\Omega$-fixing implementation of $\g$. As $\g|_{\A(I')}$ is outer, the algebra $\langle \A(I'),\Ga\rangle$ generated by $\A(I')$ and $\Ga$ is the von Neumann algebra crossed product of $\A(I')$ by $\g|_{\A(I')}$.

Now $\pi$ is graded and normal so, by choosing $\Ga_\pi$ selfadjoint, also the algebra $\langle \pi_{I'}(\A(I')),\Ga_\pi\rangle$ generated by $\pi_{I'}(\A(I'))$ and $\Ga_\pi$ is naturally isomorphic to the von Neumann algebra crossed product of $\A(I')$ by $\g|_{\A(I')}$.

So there exists an isomorphism $\Phi:\langle \A(I'),\Ga\rangle\to\langle \pi_{I'}(\A(I')),\Ga_\pi\rangle$ such that $\Phi|_{\A(I')} =\pi_{I'}$ and $\Phi(\Ga)=\Ga_\pi$, namely
\[
\Phi: x + y\Ga \mapsto \pi_{I'}(x) + \pi_{I'}(y)\Ga_\pi\ ,\quad x,y\in\A(I')\ .
\]
As $\A(I')$ is a type III factor, also $\langle \A(I'),\Ga\rangle$ is a type III factor. Thus $\Phi$ is spatial and there exists a unitary $U:\H\to\H_\pi$ such that $\Phi(x) = U x U^*$ for all $x\in\A(I')$.

Set $\rho \equiv U^*\pi(\cdot) U$. Clearly $\rho$ is a representation of $\A$ on $\H$ that is unitarily equivalent to $\pi$ and such that $\rho_{I'}$ acts identically on $\A(I')$. 

Moreover $\rho\cdot\g = \Ad\Ga\cdot \rho$. Indeed, since $ U^*\Ga_\pi= \Ga U^*$, we have
\[
\rho\cdot\g = \Ad U^*\cdot\pi\cdot\Ad\Ga 
=\Ad U^*\Ga_\pi\cdot\pi  = \Ad\Ga U^*\cdot\pi =\Ad\Ga\cdot\rho \ .
\]
With $x\in\A(I)$, we now want to show that $\r_I(x)\in\A(I)$. Indeed for all $y\in\A(I_0)$ with $\bar I_0\subset I'$ we have $[x,y]=0$, where the brackets denote the graded commutator. Therefore, choosing an interval $\tilde I\supset I'\cup I_0$, we have
\begin{equation*}
[\r_I(x),y] = [\r_{\tilde I}(x),\r_{\tilde I}(y)] = \r_{\tilde I}([x,y]) = 0\ .
\end{equation*}
It then follows by Cor. \ref{td} that $\r_I(x)\in\A(I)$.
\endproof
\subsection{$\s$-Bose and $\s$-Fermi sectors of the Bosonic subnet}
\label{sfb}
As above, let $\g$ be the vacuum preserving, involutive grading automorphism of the Fermi net
$\A$ on $S^1$ as above and $\A_b$ the fixed-point subnet. 
We denote by $\s$ a representative of the sector of $\A_b$ dual to $\g$. Choosing 
an interval $I_0\subset\mathbb R$, there is a unitary 
\[
v\in\A(I),\: v^*=v,\: \g(v)= -v \ .
\]
Then  we may take  $\s\equiv\Ad v |_{\A_b}$, so $\s$ is an automorphism of $\A_b$ localised in 
$I_0$. We have $d(\s)=1$ and $\s^2 = 1$.

Given DHR endomorphisms $\mu$ and $\nu$ of $\A_b$ 
we denote by $\e(\mu,\nu)$ the right (clockwise) statistics operator (see \cite{R1,GL2}).

The \emph{monodromy operator}  is given by
\[
m(\mu,\nu)\equiv\e(\mu,\nu)\e(\nu,\mu).
\]
Note that if $\mu$ is localised left to $\nu$, then $\e(\nu,\mu)=1$ 
and thus $m(\mu,\nu)=\e(\mu,\nu)$.

We shall also set 
\[
\k(\mu,\nu)\equiv \Phi_{\nu}(m(\nu,\mu)^*)=\Phi_{\nu}(\e(\mu,\nu)^*\e(\nu,\mu)^*)
\]
where $\Phi_{\nu}$ is the left inverse of $\nu$. 
As $\e(\mu,\nu)\in\Hom(\mu\nu,\nu\mu)$, we have $m(\mu,\nu)\in \Hom(\nu\mu,\nu\mu)$, 
therefore $\k(\mu,\nu)\in \Hom(\mu,\mu)$ and so, if $\mu$ is irreducible, 
$\k(\mu,\nu)$ is a complex number with modulus $\leq 1$.
 
Let $\mu$ be an irreducible endomorphism localised left to $\s$. 
As $\e(\mu,\s)\in\Hom(\mu\s,\s\mu)$ and $\s$ and $\mu$ commute, it 
follows that $\e(\mu,\s)$ is scalar. Denoting by $\iota$ the identity 
sector, by the braiding fusion relation we have
\[
1=\e(\mu,\iota)=\e(\mu,\s^2)
=\s(\e(\mu,\s))\e(\mu,\s)=\e(\mu,\s)\e(\mu,\s)\ ,
\]
thus $m(\mu,\s)=\e(\mu,\s)=\pm 1$.

If $\mu$ is not necessarily irreducible, we shall say that $\mu$ is 
\emph{$\s$-Bose} 
if $m(\mu,\s)= 1$ and that $\mu$ is \emph{$\s$-Fermi} 
if $m(\mu,\s)= -1$. As we have seen, if $\mu$ is irreducible then 
$\mu$ is either $\s$-Bose or $\s$-Fermi. With $S$ Rehren matrix \eqref{matS} we have:
\begin{proposition}\label{DS}
Let $\r,\nu$ be irreducible, localized endomorphisms of $\A_b$ and $\r'\equiv \r\s$. We have
$S_{\r',\nu} =\pm S_{\r,\nu}$ according with $\r$ is $\s$-Bose or 
$\s$-Fermi. 
\end{proposition}
\proof
By definition
\[
S_{\r',\nu}=S_{\r\s,\nu}=
\frac{d(\r')d(\nu)}{\sqrt{\mu_\A}}\Phi_{\nu}(\e(\r\s,\nu)^*\e(\nu,\r\s)^*)
\]
and we have $d(\r')=d(\r)$. We may also assume that $\r$, $\s$ and $\nu$ are localised one left to the next. We have
\[
m(\r',\nu)=\e(\r\s,\nu)=\r(\e(\s,\nu))\e(\r,\nu) =\pm \e(\r,\nu) = \pm m(\r,\nu)
\]
as $\e(\s,\nu)=m(\s,\nu)=\pm 1$; thus $S_{\r',\nu} =\pm S_{\r,\nu}$ where the sign depends on the $\s$-Bose/$\s$-Fermi alternative for $\nu$.
\endproof
Thus
\[
D_{\r,\nu} \equiv S_{\r,\nu} - S_{\r',\nu}= 
2\frac{d(\r)}{\sqrt{\mu_{\A_b}}}\cdot\begin{cases} 0\qquad &\text{if $\nu$ is $\s$-Bose}\\
K(\r,\nu)d(\nu)\qquad &\text{if $\nu$ is $\s$-Fermi}
\end{cases}
\]
\subsection{Graded tensor product} 
We briefly recall the notion of graded tensor product of Fermi nets. With $\A$ a Fermi net on $S^1$, denote by $\A_f(I)$ the Fermi (i.e. degree one) subspace of $\A(I)$. If $v\in\A_f(I)$ is a selfadjoint unitary then $\A(I)= \langle\A_b(I),v\rangle$ is the crossed product $\A_b(I)\rtimes \mathbb Z_2$ with respect to the action $\s =\Ad v$ on $\A_b(I)$. If $W\in\A(I)'$ is a selfadjoint unitary, $W v$ also implements $\s$ on $\A_b(I)$ and the von Neumann algebra $\langle\A_b(I),W v\rangle$ is also isomorphic to $\A_b(I)\rtimes \mathbb Z_2$, namely we have an isomorphism
\[
a\to a,\quad f\to W f, \qquad \partial a = 0, \partial f = 1\ .
\]
Let now, for $i=1,2$, $\A_i$ be a Fermi net on $S^1$ on the Hilbert 
space $\H_i$, and let $\Gamma_i$ be the associated grading unitary. 
Given an interval $I\in\I$, define the von Neumann algebra on $\H_1\otimes\H_2$
\[
\A_1(I)\hat\otimes\A_2(I)\equiv
\{a_1\otimes a_2,\  f_1\otimes 1,\  \Gamma_1\otimes f_2\}''
\]
where $a_i,f_i\in\A_i(I)$, $\partial a_i =0$ and $\partial f_i =1$. Namely $\A_1(I)\hat\otimes\A_2(I) $ is the direct sum
\[
\underbrace{\A_{1b}(I)\otimes\A_{2b}(I) + \Gamma_1\A_{1f}(I)\otimes\A_{2f}(I)}_{\rm Bosons}
+ \underbrace{\A_{1f}(I)\otimes \A_{2b}(I) + \Gamma_1 \A_{1b}(I)\otimes\A_{2f}(I)}_{\rm Fermions}
\]
Clearly the map $I\to \A_1(I)\hat\otimes\A_2(I)$ is isotone and satisfies graded
locality with respect to the grading induced by $\Gamma_1 \otimes \Gamma_2$. 
Thus it defines a Fermi net $\A_1 \hat\otimes \A_2$ on $\H_1\otimes \H_2$. 

By the previous comments, with $\A(I) = 1 
\otimes \A_2(I)$ and $W=\Gamma_1\otimes 1$, the von Neumann algebra $\hat\A_2(I)$ on 
$\H_1\otimes\H_2$ generated by $1\otimes\A_{2b}(I)$ and $\Gamma_1\otimes \A_{2f}(I)$ 
is isomorphic to $1\otimes\A_2(I)$ and hence to $\A_2(I)$. Actually, an easy 
calculation show that 
if $a, f \in \A_2(I)$, $\partial a = 0, \partial f = 1$, 
$Z\equiv \frac{1-i\Gamma_1\otimes\Gamma_2}{1-i}$ and 
 $Z_2\equiv \frac{1-i\Gamma_2}{1-i}$ then
$$(1\otimes Z_2)Z^*\big(1\otimes (a+f)\big)Z(1\otimes Z_2)^*= 1\otimes a + 
\Gamma_1\otimes f.$$
Hence the unitary operator $(1\otimes Z_2)Z^*$ on $\H_1\otimes\H_2$ implements 
the isomorphism of $1\otimes\A_2(I)$ onto $\hat\A_2(I)$ for every interval $I$. 
  
Obviously $\hat\A_1(I)\equiv \A_1(I)\otimes 1$ is isomorphic to $\A_1(I)$. Moreover 
we have 
\begin{equation}
\label{gradedtensor}
\A_1(I)\hat\otimes\A_2(I) = \hat\A_1(I)\vee\hat\A_2(I),\quad [\hat\A_1(I),\hat\A_2(I)] = 0\ ,
\end{equation}
where the brackets denote the graded commutator corresponding.

One can check that, up to isomorphism, the graded tensor product 
$\A_1(I)\hat\otimes\A_2(I)$ is the unique von Neumann algebra generated by copies
 $\hat\A_1(I)$ of $\A_1(I)$ and $\hat\A_2(I)$ of $\A_2(I)$, having 
a grading that restricts to the grading of $\hat\A_i(I)$, $i=1,2$,
satisfying the relations \eqref{gradedtensor}, and such that that $\A_{1b}(I)\vee\A_{2b}(I)$ is the usual tensor product of von Neumann algebras.

We can then define the graded tensor product of DHR representations. 
If $\lambda_1$ and $\lambda_2$ are DHR representations of $\A_1$ and $\A_2$ 
respectively and if $\lambda_1$ is graded, then it can be shown that there 
exists a (necessarily unique) representation $\lambda_1 \hat\otimes \lambda_2$ on 
$\H_{\lambda_1} \otimes \H_{\lambda_2}$ such that, for every interval 
$I\subset S^1$,  
\[
(\lambda_1 \hat\otimes \lambda_2)_I (x_1\otimes a_2+
\Gamma_1x_1 \otimes f_2) = {\lambda_1}_I(x_1)\otimes {\lambda_2}_I(a_2)
+\Gamma_{\lambda_1}{\lambda_1}_I(x_1)\otimes {\lambda_2}_I(f_2)
\]
where $x_1\in \A_1(I)$, $a_2,f_2\in\A_2(I)$, $\partial a_2 =0$ and $\partial 
f_2=1$. 
Analogously we can 
define the graded tensor product of general representations defined below, 
provided one of them is graded.
\section{Nets on $\mathbb R$ and on a cover of $S^1$}
Besides nets on $S^1$, it will be natural to consider nets on $\mathbb R$ and nets on topological covers of $S^1$. Indeed the two notions are related as we shall see.
\subsection{Nets on $\mathbb R$}
\label{netR}
Denote by $\I_{\mathbb R}$ the family of open intervals of $\mathbb R$, 
i.e. of open, non-empty, connected, bounded subsets of $\mathbb R$. 

The M\"obius group $\Mob$ can be naturally viewed as a subgroup of $\Diff(S^1)$. We then have a corresponding inclusion $\Mob^{(n)}\subset\Diff^{(n)}(S^1)$ of covering groups. In the following we denote by $G$ a group that can be either the M\"obius group $\Mob$  or the diffeomorphism group $\Diff(S^1)$. Analogously, we have $G^{(n)}=\Mob^{(n)}$ or
$G^{(n)}=\Diff^{(n)}(S^1)$.
By identifying $\mathbb R$ with $S^1\!\setminus\!\{-1\}$ via the stereographic 
map, we have a local action of $G$ on $\mathbb R$
(where $SL(2,\mathbb R)/\{1,-1\}\simeq\Mob$ acts on $\mathbb R$ by linear fractional transformations). See \cite{BGL,GL5} for the definition and  a discussion about local actions. 

A \emph{ $G$-covariant net on $\mathbb R$} (of von Neumann algebras) $\A$ is a isotone map 
\[
I\in\I_{\mathbb R}\mapsto\A(I)
\]
that associates to each $I\in\I_{\mathbb R}$ a von Neumann algebra $\A(I)$ on a fixed Hilbert space $\H$ and there exists a projective, positive energy, unitary representation $U$ of $G^{(\infty)}$ on 
$\H$ with
\[
U(g)\A(I)U(g)^{-1} = \A(\dot{g} I), \quad g\in\U_I,
\]
for every $I\in\I_{\mathbb R}$ where $\U_I$ is the connected component of the identity in $G^{(\infty)}$ of the open set $\{g\in G^{(\infty)}: gI\in\I_{\mathbb R}\}$.

We will further assume the irreducibility of $\A$ and the existence of a vacuum vector $\Omega$ for $U$, although this is not always necessary.

Note that it would be enough to require the existence of the projective unitary representation $U$ only in a neighbourhood of the identity of $G^{(\infty)}$, then $U$ would extend to all $G^{(\infty)}$ by multiplicativity.

Since the cohomology of the Lie algebra of $\Mob$ is trivial, 
by multiplying $U(g)$ by a phase factor (in a unique fashion), we may remove the 
projectiveness of $U|_{\Mob^{(\infty)}}$ and assume that the restriction of $U$ to $\Mob^{(\infty)}$
is a unitary representation of $\Mob^{(\infty)}$. 
\subsection{Nets on a cover of $S^1$}
\label{coverNets}
The group $G^{(n)}$ has a natural action on $\Sn$, $n$ finite or infinite, the one obtained by promoting the action of $G$ on $S^1$, see Sect. \ref{top}. Here $\Sn$ denotes the $n$-cover of $S^1$. Denote by $\I^{(n)}$ the family of intervals of $\Sn$, i.e. $I\in\I^{(n)}$ iff $I$ is a connected subset of $\Sn$ that projects onto a (proper) interval of the base $S^1$.

A \emph{$G$-covariant net $\A$ on $\Sn$} is a isotone map 
\[
I\in\I^{(n)}\mapsto\A(I)
\]
that associates with each $I\in\I^{(n)}$ a von Neumann algebra $\A(I)$ on a fixed Hilbert space $\H$, and there exists a projective
unitary, positive energy representation $U$ of $G^{(\infty)}$ on 
$\H$, implementing a  covariant action on $\A$, namely 
\[
U(g)\A(I)U(g)^{-1} = \A(\dot{g}I),\quad  I\in \I^{(n)},\ 
g\in G^{(\infty)}
\]
Here $\dot{g}$ denotes the image of $g$ in $G^{(n)}$ under the quotient map.

As above, we may also assume irreducibility of the net and the existence of a vacuum vector $\Omega$ for $U$.

Of course a $G$-covariant net $\tilde\A$ on $\Sn$ determines a $G$-covariant net $\tilde\A$ on $\mathbb R$. Indeed, with $p:\Si\to S^1$ the covering map, every connected component of 
$p^{-1}(S^{1}\!\setminus\!\{{\rm point}\})$
 is a copy of 
$\mathbb R$ in $\Sn$ and we may restrict $\tilde\A$ to any of this copy; by $G$-covariance we get always the same $G$-covariant net on $\mathbb R$, up to unitary equivalence.

Conversely, a $G$-covariant net $\A$ on $\mathbb R$ can be extended to a net
$\tilde\A$ on $\Si$ by defining, for any given $I\in\I$,
\begin{equation}\label{gext}
\tilde\A(I)\equiv U(g)\A(I_1)U(g)^{-1}
\end{equation}
where $I_1\in\I_{\mathbb R}$, $g\in G^{(\infty)}$, and $gI_1 = I$. Here 
the action of $G^{(\infty)}$ on $\Si$ is the one obtained by 
promoting the action of $G^{(\infty)}$ on $S^1$ (see above). 
Therefore we have:
\begin{proposition}\label{rtos}
There is a 1-1 correspondence (up to unitary equivalence) between $G$-covariant nets on $\mathbb R$ and $G$-covariant  nets on $\Si$
\end{proposition}
\proof
We only show that if $\A$ is a $G$-covariant net on $\mathbb R$, then formula \eqref{gext} well defines $\tilde\A(I)$. If $g'\in G^{(\infty)}$ also satisfies $g'I_1 = I$, then $g$ and $g'$ have ``the same degree", namely $h\equiv g^{-1}g'$ maps $I_1$ onto $I_1$ and is in $\U_{I_1}$. 
Then $U(h)\A(I_1)U(h)^{-1}= \A(I_1)$, so $U(g)\A(I_1)U(g)^{-1}= U(g')\A(I_1)U(g')^{-1}$. 
The rest is clear.
\endproof
With $\A$ a be a $G$-covariant net on $\mathbb R$ as above, we shall say that a unitary operator $V$ on $\H$ is a gauge unitary if 
\begin{equation}
\label{gauge}
V\A(I)V^* =\A(I), \ VU(g)=U(g)V,
\end{equation}
for all $I\in\I_{\mathbb R}$ and $g$ in a suitable neighbourhood of 
the identity of $\Mob$. Clearly the same relations \eqref{gauge} then 
hold true for all $I\in\Si$, $g\in\Mobi$ for the corresponding net 
on $\Si$. Now ${\rm rot}_{2\pi}$ is a 
central element of $G^{(\infty)}$ and it is immediate to check that if $n\in\mathbb 
N$: 
\[
U(2\pi)^n\ \text{is a gauge unitary}\Leftrightarrow\text{$\A$ extends to a 
$G$-covariant net on $\Sn$} ;
\]
then $\A$ is covariant with respect to the corresponding action of 
$G^{(\infty)}$ on $\Sn$. Clearly the net on $\Sn$ is the projection of the 
net on $\Si$ by the covering map from $\Si$ to $\Sn$. In other words we have the following where $n\in\mathbb N$:
\begin{corollary}
A $G$-covariant net $\A$ on $\mathbb R$ is the restriction of a $G$-covariant net on $\Sn$ if and only if $U(2n\pi)$ is a gauge unitary for $\A$. This is the case, in particular, if the representation $U$ of $G^{(\infty)}$ factors through a representation of $G^{(n)}$ (i.e. $U(2n\pi) = 1$).
\end{corollary}
\noindent
If $\A$ is a $G$-covariant net on $\mathbb R$ we shall denote by $\A^{(\infty)}$ its promotion to $\Si$. If $U(2\pi)^n$ is a gauge unitary for some $n\in\mathbb N$, we shall denote by $\A^{(n)}$ the promotion of $\A$ to $\Sn$.

We now define the \emph{promotion} of a net $\A$ on $S^1$ to a net $\A^{(n)}$ on $\Sn$. If $\A$ is a $G$-covariant net on $S^1$, then its restriction $\A_0$ to $\mathbb R$ is a $G$-covariant net on $\mathbb R$, and we set
\[
\A^{(n)}\equiv \A_0^{(n)}\ ;
\]
here, if $n$ is finite, we have to assume that $U(2\pi n)$ is a gauge unitary.
We shall be mainly interested in the case of a $G$-covariant Fermi net $\A$ on $S^1$. As $U(4\pi)=1$ in this case, we have a natural net $\A^{(2)}$ on $\S2$ associated with $\A$. Of course if $\A$ is local then $U(2\pi) =1$ the net $\A^{(2)}$ on $\S2$ is defined in this case. 

Of course if $\A^{(n)}$  is the promotion to $\Sn$ of a net $\A$ on $S^1$, then  $\A^{(n)}(I)=\A(pI)$ for any interval $I$ of $\Sn$.
%


%
\section{Solitons and representations of cover nets}
We now consider more general representations associated with a Fermi conformal net. The point is that a Fermi conformal nets lives naturally in a double cover of $S^1$ rather than on $S^1$ itself  because the $2\pi$ rotation unitary $U(2\pi)$ is not the identity ($U(2\pi)=\Gamma$, see Sect. \ref{vacss}) but $U(4\pi)=1$. So representations as a net on $\S2$ come naturally into play. These representations can be equivalently viewed as a natural class of solitions.
\subsection{Representations of a net on $\Sn$}
We begin by giving the notion of representation for a net on a cover of $S^1$.

Let $\A$ be a $G$-covariant net of von Neumann algebras on $\Sn$ and $U$ the associated covariance unitary representation of of $G^{(\infty)}$ (thus $U(2\pi n)$ is a gauge unitary). A \emph{representation} of $\A$ is a map
\[
I\in\I^{(n)}\mapsto \l_I
\]
where $\l_I$ is a normal representation of $\A(I)$ on a fixed Hilbert space $\H_\l$ with the usual consistency condition $\l_{\tilde I}|_{\A(I)} = \l_I$ if $\tilde I\supset I$. 

We shall say that $\l$ is \emph{$G$-covariant} if there exists a projective unitary, positive energy representation $U_\l$ of $G^{(\infty)}$ on $\H_\l$ such that
\[
U_\l(g)\l_I(x)U_\l(g)^* = \l_{\dot{g}I}(U(g)xU(g)^*),\ g\in G^{(\infty)},\ I\in\I^{(n)}\ .
\]
Here $\dot g$ denotes the image of $g\in G^{(\infty)}$ in $G^{(n)}$ under the quotient map, thus $\dot{g}I$ is the projection of $gI\in\Si$ onto $\Sn$.
 
Let now $\A$ be a $G$-covariant net on $S^1$. If $\l$ is a representation of $\A$, then $\l$ promotes to a representation $\l^{(n)}$ of $\A^{(n)}$  given by
\[
\l^{(n)}_I \equiv \l_{p(I)},\quad I\in\I^{(n)}\ .
\]
If $\l$ is $G$-covariant and $U_\l$ is the associated unitary representation of $G^{(n)}$, then $\l^{(n)}$ is also $G$-covariant with the same unitary representation $U_\l$ of $G^{(n)}$.
\begin{lemma}\label{prom}
Let $\A$ a $G$-covariant Fermi (resp. local) net on $S^1$ and $\nu$ a $G$-covariant representation of $\A^{(n)}$, with $U_\nu$ the associated projective unitary representation of $G^{(n)}$.

Then $\nu$ is the promotion $\l^{(n)}$ of a $G$-covariant representation $\l$ of $\A$ iff $U_\nu(2\pi)$ implements the grading (resp. the identity) in the representation $\nu$. 
\end{lemma}
\proof
We assume $\A$ to be a Fermi net (the local case is simpler and follows by the same argument). If $\l$ is a $G$-covariant representation of $\A$, then $U_\l(2\pi)$ implements the grading by the spin-statistics theorem \cite{GL2} (applied to $\l |_{\A_b}$). Conversely, if $\nu$ is a $G$ covariant representation of $\A^{(n)}$ and $U_\nu(2\pi)$ implements the grading, let $\l_0$ be the restriction of $\nu$ to a copy of $\mathbb R$ in $S^{1(n)}$. Then $\l_0$ extends by $G$-covariance to a representation of $\A$ (Prop. \ref{rtos}). 
\endproof
\subsection{Solitons}
Let $\A_0$ be a net of von Neumann algebras on $\mathbb R$ and denote by $\bar\I_{\mathbb R}$ the family of intervals/half-lines of $\mathbb R$ (open connected subsets of $\mathbb R$ different from $\emptyset , \mathbb R$). If $I$ is a half-line, we define $\A(I)$ as the von Neumann algebras generated by the von Neumann algebras associated to the intervals contained in $I$.

A \emph{soliton} $\l$ of $\A_0$ a map
\[
I\in\bar\I_{\mathbb R}\mapsto \l_I
\]
where $\l_I$ is a normal representation of the von Neumann algebra $\A_0(I)$ on a fixed Hilbert space $\H_\l$ with the usual isotony property $\l_{\tilde I}|_{\A(I)} =\l_I$, $I\subset\tilde I$.

If $\A$ is a net on $S^1$, by a soliton of $\A$ we shall mean a soliton of the restriction of $\A$ to $\mathbb R$. 

Let $\A$ be a conformal net on $S^1$. We set
\[
\text{$\A_0\ \equiv\ $ restriction of $\A$ to $\mathbb R$}. 
\]
If $\l$ is a DHR representation of $\A$ then, obviously, $\l |_{\A_0}$ is a soliton of $\A_0$. We shall say that a soliton $\l_0$ of $\A$ is a DHR representation of $\A$ if it arises in this way, namely $\l_0=\l |_{\A_0}$ with $\l$ a DHR representation of $\A$ (in other words, $\l |_{\A_0}$ extends to a DHR representation of $\A$, note that the extension is unique if strong additivity is assumed).

More generally, we get solitons of $\A$ by restricting to $\A_0$ representations of the cover nets $\A^{(n)}$.

Let $\A$ be a $G$-covariant net on $S^1$ with covariance unitary representation $U$. A \emph{$G$-covariant soliton} $\l$ of $\A$ is a soliton of $\A_0$ such that there exists a projective unitary representation $U_\l$ of $G^{(\infty)}$ such that for every bounded interval $I$ of $\mathbb R$ we have
\[
\l_{\dot g I}(U(g)xU(g)^*) = U_\l (g)\l_I(x)U_\l (g)^*\ ,\quad  g\in\U_I ,\ x\in\A(I)\ ,
\]
where $\U_I$ is a connected neighborhood of the identity of $G^{(\infty)}$ as in Sect. \ref{netR} and $\dot g$ is the image of $g$ in $G$.
\begin{proposition}\label{r-si}
Let $\A$ be a $G$-covariant net on $S^1$.
There is a one-to-one correspondence between
\begin{itemize}
\item[$(a)$] $G$-covariant representations of $\A^{(\infty)}$,
\item[$(b)$] $G$-covariant solitons of $\A$.
\end{itemize}
The correspondence is given by restricting a representation of $\A^{(\infty)}$ to a copy of $\A_0$.

Suppose $\A$ is local. Then the above restricts to a one-to-one correspondence between $(a)$: $G$-covariant representations of $\A^{(n)}$ and $(b)$: $G$-covariant solitons of $\A$ with $U_\l(2\pi n)$ commuting with $\l$.

Suppose $\A$ is Fermi. Then the above restricts to a one-to-one correspondence between $(a)$: $G$-covariant representations of $\A^{(n)}$ and $(b)$: $G$-covariant solitons of $\A$ with $U_\l(2\pi n)$ implementing the grading ($n$ even) or commuting with $\l$ ($n$ odd).
\end{proposition}
\proof
Clearly if $\l$ is a $G$-covariant representation of $\A^{(\infty)}$ then $\l_0\equiv\l |_{\A_0}$ is a $G$-covariant solitons of $\A$. Conversely, if $\l_0$ is $G$-covariant solitons of $\A$ then we set
\[
\l_{g I}(U(g)xU(g)^*) = U_\l (g)\l_I(x)U_\l (g)^*\ ,\quad  g\in G^{(\infty)} ,\ x\in\A(I)\ ,
\]
where $I\subset\mathbb R$ is a bounded interval for any given copy of $\mathbb R$ is in $\Si$ and $G^{(\infty)}$ acts on $\Si$ as usual . By  $G$-covariance the above formula well-define a representation of $\A^{(\infty)}$.

The second part follows because by the vacuum spin-statistics relation we have $U(2\pi)=1$ in the local case and $U(2\pi)=\Gamma$ in the Fermi case.
\endproof
A \emph{graded soliton} of the Fermi net $\A$ is, of course, a soliton such that the grading is unitarily implemented, namely the soliton $\l$ is graded iff there exists a unitary $\Ga_\l$ on $\H_\l$ such that
\[
\l_I(\g(x)) = \Ga_\l \l(x)\Ga^*_\l,\quad I\in\I_{\mathbb R},\ x\in\A(I).
\]
\subsection{Neveu-Schwarz and Ramond representations}
We give now the general definition for a representation of a Fermi net. We have two formulations for this notion: as a representation of the cover net and as a soliton.

Let $\A$ be a Fermi net on $S^1$. A \emph{general representation} $\l$ of $\A$ is a representation the cover net of $\A^{(\infty)}$ such that $\l$ restricts to a DHR representation $\l_b$ of the Bose subnet $\A_b$.

We shall see here later on that a general representation is indeed a representation of $\A^{(2)}$. We begin by noticing an automatic covariance for a general representation $\l$. In this case $\l$ is not always graded.
\begin{proposition}\label{diffcov0}
Let $\A$ be a Fermi conformal net on $S^1$. Every general representation $\l$ of $\A$ is diffeomorphism covariant with positive energy. Indeed $U_\l = U_{\l_b}$.
\end{proposition}
\proof
The proof follows by an immediate extension of the argument proving Prop. \ref{DHRgraded}.
\endproof
Let $\A$ be a Fermi net on $S^1$. A \emph{general soliton} $\l$ of $\A$ is a  diffeomorphism covariant soliton of $\A$ such that $\l$ restricts to a DHR representation $\l_b$ of the Bose subnet $\A_b$. (In other words $\l|_{\A_{b,0}}$ extends to a DHR representation $\l_b$ of $\A_b$.)

By Prop. \ref{r-si} we have a one-to-one correspondence between
\begin{gather*}
\text{General Representations}\\
\Updownarrow\\
\text{General Solitons}
\end{gather*}
We prove now the diffeomorphism covariance of general solitons in the strong additive case. Notice that strong additivity is automatic in the completely rational case \cite{LX}. The general proof follows by Prop. \ref{acov} below.
\begin{proposition}\label{diffcov}
Let $\A$ be a  Fermi conformal net on $S^1$. Every soliton $\l$ of $\A$ such that $\l_b$ is a DHR representation is diffeomorphism covariant with positive energy, namely $\l$ is a general soliton. Indeed $U_\l = U_{\l_b}$.
\end{proposition}
\proof
(Assuming strongly additive). Note that if $g\in\Diff(S^1)$ is localized in an interval $I$ then, up to a phase, $U_{\l_b}(g) = \l_b(U(g))\in\A_b(I)$.

Let $I$ be an interval contained in $\mathbb R = S^1\setminus\{-1\}$ a suppose that $g\in\Diff(S^1)$ is localised in $I$; then
\begin{multline*}
U_{\l_b}(g)\l_{\tilde I}(x)U^*_{\l_b}(g) = {\l_b}_{\tilde I}(U(g))\l_{\tilde I}(x){\l_b}_{\tilde I}(U^*(g))
= {\l}_{\tilde I}(U(g))\l_{\tilde I}(x){\l}_{\tilde I}(U^*(g))\\
= \l_{\tilde I}\big(U(g)xU(g)^*\big)
= \l_{g\tilde I}\big(U(g)xU(g)^*\big)\ ,
\quad x\in\A(\tilde I)\ ,
\end{multline*}
for any interval $\tilde I\supset I$ of $\mathbb R$. 

Notice at this point that if $x\in\A(\tilde I)$ and $v\in\A_b(\tilde I')$ then ${\l_b}_{\tilde I'}(v)$ commutes with $\l_{\tilde I}(x)$:
\[
[{\l_b}_{\tilde I'}(v),\l_{\tilde I}(x)]=0\ .
\]
Indeed, as $\A$ is assumed to be strongly additive, also $\A_b$ is strongly additive \cite{X01}, see also \cite{L03}, so it is sufficient to show the above relation with $v\in\A_b(I_0)$,  where $I_0$ is any interval with closure contained in $\tilde I'\setminus \{-1\}$. Then
\begin{equation*}
[{\l_b}_{\tilde I'}(v),\l_{\tilde I}(x)] =[{\l_b}_{I_0}(v),\l_{\tilde I}(x)]
= [{\l}_{I_0}(v),\l_{\tilde I}(x)] 
= [{\l}_{I_1}(v),\l_{I_1}(x)] 
= {\l}_{I_1}([v,x]) = 0\ ,
\end{equation*}
where we have taken an interval $I_1$ of $\mathbb R$ containing $I_0\cup\tilde I$.

Let now $\U_I$ be a connected neighbourhood of the identity in $\Diff(S^1)$ such that $gI$ is an interval of $\mathbb R$ for all $g\in\U_I$. Take $g\in\U_I$ and let $\tilde I$ be an interval of $\mathbb R$ with $\tilde I\supset I\cup gI$. Let $h\in \Diff(S^1)$ be a diffeomorphism localised in an interval of $\mathbb R$ containing $\tilde I$ such that $h|_{\tilde I} = g|_{\tilde I}$. Then
\[
U(g) = U(h)v
\]
where $v\equiv U(h^{-1}g )\in \A_b(\tilde I')$ (the Virasoro subnet is contained in the Bose subnet).

We then have with $x\in\A(\tilde I)$:
\begin{multline*}
U_{\l_b}(g)\l_{\tilde I}(x)U^*_{\l_b}(g) =
U_{\l_b}(h){\l_b}_{\tilde I'}(v)\l_{\tilde I}(x){\l_b}_{\tilde I'}(v^*)U^*_{\l_b}(h) \\
= U_{\l_b}(h)\l_{\tilde I}(x)U^*_{\l_b}(h)\ .
=\l_{h\tilde I}\big(U(h)xU(h)^*\big)
=\l_{g\tilde I}\big(U(g)xU(g)^*\big)
\end{multline*}
\endproof
\begin{proposition}\label{n1}
Let $\A$ be a Fermi conformal net on $S^1$ and $\l$ an irreducible 
general soliton of $\A$. With $\l_{b,0}\equiv\l |_{\A_{b,0}}$, the following are equivalent:
\begin{itemize}
\item[$(i)$] $\l$ is graded,
\item[$(ii)$] $\l_{b,0}$ is not irreducible,
\item[$(iii)$] $\l_{b,0}$ is direct sum of two inequivalent irreducible representations 
of $\A_{b,0}$:
\[
\l_{b,0} \simeq \r\oplus \r'
\]
where $\r$ is an irreducible DHR representation of $\A_{b,0}$ and 
$\r'\equiv \r\s$ with $\s$ is the dual involutive localised automorphism of $\A_b$ as above. \end{itemize}
If the above holds and $\l$ (i.e. $\r$) has finite index
we then have equation for the statistics phases
\begin{equation}\label{mp1}
\omega_{\r'} = - m(\r,\s)\omega_\r
\end{equation}
Thus
\begin{equation}\label{mp2}
U_{\r'}(2\pi) = \mp U_\r(2\pi)
\end{equation}
according to $\r$ is $\s$-Bose/$\s$-Fermi.    
\end{proposition}
\proof
$(i)\Rightarrow (iii)$: If $\l$ is graded, then the unitary $\Ga_\l$ implementing the grading in the representation commutes with $\l(\A_{b,0})$, so $\l_{b,0}$ is reducible. Moreover, as $\l(\A_{b,0})''= \l(\A_0)''\cap\{\Ga_\l\}'$ , we have
$\l(\A_{b,0})'= \l(\A_0)'\vee\{\Ga_\l\}''=\{\Ga_\l\}''$, so $\l(\A_{b,0})'$ is 2-dimensional. 
Let $\r$ one of the two irreducible components of $\l_{b,0}$. As the dual canonical endomorphism associated with $\A_{b,0}\subset\A_0$ is $\iota\oplus\s$ (see \cite{LR1}), the other component must be $\r'\equiv\r\s$, namely $\l_{b,0}=\r\oplus\r'$.

In order to show eq. \eqref{mp2} we may assume that $\r$ is localized left to $\s$, so $\r$ and $\s$ commute. By using the cocycle equations for the statistics operators, we then have
\begin{multline*}
\e(\r',\r') = \e(\r,\r')\r(\e(\s,\r')) 
= \e(\r,\r')\e(\s,\r') 
= \r(\e(\r,\s))\e(\r,\r)\r(\e(\s,\s))\e(\s,\r) \\
= \e(\r,\s)\e(\r,\r)\e(\s,\s)\e(\s,\r)
= - m(\r,\s)\e(\r,\r)\ ,
\end{multline*}
where we have used that $\e(\s,\r')$ and $\e(\r,\s)$ are scalars and  $\e(\s,\s)= -1$. So, if $\r$ has finite index, we infer the equation \eqref{mp1} for the statistics phases and eq. \eqref{mp2} follows by the conformal spin-statistics theorem \cite{GL2}.

The implication $(iii)\Rightarrow (ii)$ is obvious. For the implication $(ii)\Rightarrow (i)$ assume that $\l$ is not graded, namely $\g$ is not unitarily implemented in the representation $\l$; then $\l$ is irreducible by known arguments, see \cite{Ro} (the fixed point of an irreducible $C^*$-algebra with respect to a period two outer automorphism is still irreducible).
\endproof
\noindent
{\it Remark.} In the above proposition the diffeomorphism covariance of $\l$ is unnecessary, only the M\"obius covariance of $\l_b$ has been used.
\begin{corollary}\label{extrep}
Let $\A$ be a Fermi conformal net on $S^1$ and $\l$ an irreducible general soliton of $\A$ with finite index.
The following alternative holds:
\begin{itemize}
\item[$(a)$] $\l$ is a DHR representation of $\A$. Equivalently $U_{\l_b}(2\pi)$ is not a scalar.
\item[$(b)$] $\l$ is the restriction of a representation of $\A^{(2)}$ and $\l$ is not a DHR representation of $\A$. Equivalently $U_{\l_b}(2\pi)$ is a scalar.
\end{itemize}
\end{corollary}
\proof
First note that $U_{\l_b}(4\pi) = U_{\r}(4\pi)\oplus U_{\r'}(4\pi)$ is a scalar by eq. \eqref{mp2}, so $\l$ is always representation of $\A^{(2)}$ by Prop. \ref{r-si}.

If $\l$ is a DHR representation of $\A$ then by, Prop. \ref{DHRgraded}, $U_\l(2\pi) =\Gamma$ is not a scalar.

Assume now that $\l$ is not a DHR representation of $\A$. It remains to show that if $U_\l(2\pi)$ is a scalar. 

If $\l$ is not graded, then by Prop. \ref{n1} $\l_b$ is irreducible; as $U_\l(2\pi)$ commutes with $\l_b$, must then be a scalar.

Finally, if $\l$ is graded, then $\r$ must be a $\s$-Fermi sector, as otherwise its $\alpha$-induction $\alpha^+_\r$ to $\A$ would be a DHR representation of $\A$, while $\alpha^+_\r= \l$ by Frobenious reciprocity. Then by eq.  \eqref{mp2} we have $U_{\r'}(2\pi) = U_\r(2\pi)$, namely $U_\l(2\pi)$ is a scalar.
\endproof
{\it Remark.} The finite index assumption in the above proposition is probably unnecessary (it has been used to make use of eq.  \eqref{mp2}).
\medskip

\noindent
It is now convenient to use the following terminology concerning the general soliton in Corollary \ref{extrep}. In the first case, namely if $\l$ is a DHR representation of $\A$, we say that $\l$ is a \emph{Neveu-Schwarz representation} of $\A$. In the second case, namely if $\l$ is a soliton (and not a representation of $\A$) we say that $\l$ is a \emph{Ramond representation} of $\A$. 

If $\l$ is reducible we shall say that $\l$ is a Neveu-Schwarz representation of $\A$ if $\l$ is a DHR representation of $\A$; we shall say that $\l$ is a Ramond representation of $\A$ if no subrepresentation of $\l$ is a DHR representation of $\A$.

Let as above $\A$ be a Fermi conformal net on $S^1$ and $\A_b$ the Bose subnet. We shall now study the representations of $\A_b$ in relation to the representations of $\A$.

Given a representation $\nu$ of $\A_b$ we consider its $\a$-induction $\a_\nu$ to $\A$ (say right $\a$-induction, so $\a_\nu\equiv\a^+_\nu$). More precisely we first restrict $\nu$ to the net $\A_{b,0}$ on the real line and then consider its $\a$-induction $\a_\nu$ which is a soliton of $\A_0$. Note that $\a_\nu$ is defined in \cite{LR1} assuming Haag duality on the real line (equivalent to strong additivity in the conformal case), but this assumption is unnecessary here because we are considering nets on $S^1$ that satisfy Haag duality on $S^1$, see \cite[Sect. 3.1]{LX}.
\begin{proposition}\label{acov}
If $\nu$ is a DHR irreducible representation of $\A_b$ then $\a_\nu$ is a general soliton of $\A$, namely $\a_\nu$ is diffeomorphism covariant with positive energy. We have:
\begin{gather*}
\text{$\nu$ is $\s$-Bose} \Leftrightarrow \text{$\a_\nu$ is Neveu-Schwarz}\\
\text{$\nu$ is $\s$-Fermi} \Leftrightarrow \text{$\a_\nu$ is Ramond}
\end{gather*}
\end{proposition}
\proof
We only sketch the proof.
We use the extension method by Roberts (see \cite{Ro}), but in a covariant way.
As $\nu$ is diffeomorphism covariant, there is a localized unitary cocycle $w$ with value in $\A_b$ such that
\begin{equation}\label{ce}
 \Ad w_g\cdot \nu = \Ad U(g)\cdot \nu\cdot\Ad U(g^{-1}),\quad g\in\Diff(S^1),
\end{equation}
see \cite[Sect. 8]{GL1}. The cocycle $w$ reconstructs $\nu$. Now $w$ can be seen as a cocycle with values in $\A$ but with less localization properties (if $w$ is bi-localized in $I\cup gI$ in $\A_b$, it is in general only localized in $\tilde I$ in $\A$ where $\tilde I$ is an interval containing $I\cup gI$). Thus, in general,  $w$ can be associated with a soliton of $\A$, which is $\a_\nu$. Now the covariance equation still remains true for $g$ in a neighborhood of the identity, and gives the diffeomorphism covariance of $\a_\nu$.

If $\nu$ a DHR irreducible representation of $\A_b$, its $\a$-induction to $\A$ is a DHR representation iff $\nu$ has trivial monodromy with the dual canonical endomorphism $\iota\oplus\s$ of $\A_b$ (the restriction to $\A_b$ of the identity representation of $\A$), thus iff $\nu$ has trivial monodromy with $\s$. By definition this means that $\nu$ is a $\s$-Bose representation.
\endproof
An immediate corollary is the following.
\begin{corollary}
A  DHR representation $\nu$ of $\A_b$ is $\s$-Bose (resp. $\s$-Fermi) if and only if $\nu$ is the restriction of a Neveu-Schwarz (resp. Ramond) representation of $\A$.
\end{corollary}
\noindent
Because of the above corollary, we shall also call Neveu-Schwarz (resp. Ramond) representation of $\A_b$ a $\s$-Bose (resp. $\s$-Fermi) representation of $\A_b$.
\subsection{Tensor categorical structure. Jones index}
In the local case, the DHR argument shows that every representation is equivalent to a localized endomorphism and a first basic consequence of this fact is that representations give rise to a tensor category because one can indeed compose localized endomorphisms and get the monoidal structure. As we have seen in Prop. \ref{dhrendo}, we can extend the DHR argument to the case of representations of a Fermi net; this argument however makes use that DHR representations are automatically graded. In order to have a tensor structure for general representations of a Fermi net, we consider graded representations only. Note however that if $\l$ is any soliton, setting where $\l_\g\equiv\g\l\g^{-1}$, then $\l\oplus\l_\g$ is graded.
We obviously have
\[
{\rm Sol}_\g(\A)\subset{\rm Sol}(\A)\ .
\]
where ${\rm Sol}(\A)$ denotes the general solitons of $\A$ that are localized, say, in a right half line of $\mathbb R$ and ${\rm Sol}_\g(\A)$ are the general solitons commuting with the grading, namely $\l\in{\rm Sol}(\A)$ belongs to  ${\rm Sol}_\g(\A)$ iff $\l=\l^\g$.

Now every general representation of $\A$ is equivalent to an element of ${\rm Sol}(\A)$. Thus, if $\l_1,\l_2\in{\rm Sol}(\A)$ an intertwiner $T:\l_1\to\l_2$ is a bounded linear operator such that
\[
T\l_1(x)=\l_2(x)T,\quad x\in\A_0\ .
\]
Note that $T$ is not necessarily localized in a half-line, but by twisted duality $T\in Z\A(I)Z^*$ where $I$ is a half-line where $\l_1$ and $\l_2$ are both localized. By further assuming that $T$ commutes with $\Gamma$, we also have that  $T$ commutes with $Z$, so $T\in\A(I)$ and $\partial T = 0$, thus $T\in\A_b(I)$.

It follows that ${\rm Sol}_\g(\A)$ is a tensor category where the arrows are defined by
\[
{\rm Hom}(\l_1,\l_2)\equiv T:\l_1\to\l_2,\ T\Gamma =\Gamma T\ .
\]
\begin{proposition}
$\a$-induction is a surjective tensor functor 
\[
{\rm Rep}(\A_b)\overset{\a}{\longrightarrow}{\rm Sol}_\g(\A)\ .
\]
where ${\rm Rep}(\A_b)$ is the tensor category of endomorphisms of $\A_b$ localized in intervals of $\mathbb R$. The kernel of $\a$ is $\{\iota,\s\}$.
\end{proposition}
\proof
First of all note that if $\nu\in{\rm Rep}(\A_b)$ is localized in the interval $I$, then $\a_\nu$ commutes with $\g$. Indeed both $\a_\nu$ and $\a^\g_\nu= \g\a_\nu\g$ have the same restriction to $\A_b$. If $v\in\A(I)$ is a selfadjoint unitary with $\partial v = 1$, then $\a_\nu(v) = uvu^*$ for a Bose unitary $u$ in the covariance cocycle for $\nu$, therefore 
\begin{equation*}
\a^\g_\nu(v) =\g(\a_\nu(\g(v)))= -\g(\a_\nu(v))
= -\g(uvu)  = -\g(u)\g(v)\g(u) = uvu =\a_\nu(v)\ ,
\end{equation*}
so $\a=\a^\g_\nu$ because $\A$ is generated by $\A_b$ and $v$.

Conversely, if $\l\in{\rm Sol}_\g(\A)$, then $\l_b$ is a localized endomorphism of $\A_b$ because $\l$ commutes with $\g$. Now the covariance cocycle $w_g$ for $\l$ has degree $0$; indeed $\g(w_g)$ is also a covariance cocycle for $\l$ thus $\g(w_g)=w_g$ up to a phase that must be 1 by the cocycle property. Thus $w$ belongs to the Bose subnet and is the covariance cocycle for $\l_b$; a calculation as above then shows that $\l = \a_{\l_b}$.

It is now rather immediate to show that $\alpha_\nu$ is the identity iff $\nu =\iota$ or $\nu=\s$ and that we have a corresponding surjective map of the arrows.
\endproof
So we have a braided tensor category ${\rm Sol}_\g(\A)$. In particular, the \emph{Jones index} of $\l$ is defined, as in the local case, see \cite{L5,L5'}. 

Note however that an irreducible element of ${\rm Sol}_\g(\A)$ may be reducible as a general representation of $\A$, namely it can decomposed into non-graded subrepresentations.
\section{Modular and superconformal invariance}
We now study Fermi nets whose Bose subnets are modular and get some consequences in the supersymmetric case.
\subsection{Modular nets}
We recall here a discussion made in \cite{KL05}.
Let $\B$ be a completely rational local conformal net of von Neumann algebras on $S^1$. Then the tensor category or representations of $\B$ is modular, i.e. rational with non-degenerate unitary braiding  \cite{KLM}.

We then have Rehren \cite{R1} matrices defining
a unitary representation of the group $SL(2,{\mathbb Z})$ on the space spanned by the irreducible sectors (i.e. unitary equivalence classes of representations) $\r$'s , in particular we have the matrices $T \equiv \{T_{\l,\nu}\}$ and $S \equiv \{S_{\l,\nu}\}$ where
\begin{equation}\label{matS}
S_{\l,\nu}\equiv 
\frac{d(\l)d(\nu)}{\sqrt{\sum_i d(\r_i)^2}}\Phi_{\nu}(\e(\nu,\l)^*\e(\l,\nu)^*) \ .
\end{equation}
Here $\e(\nu,\l)$ is the right statistics operators between $\l$
and $\nu$, $\Phi_{\nu}$ is left inverse of $\nu$ and $d(\l)$ is the 
statistical dimension of $\l$.  

Assume that $\B$ is diffeomorphism covariant with central charge $c$.
For a sector $\r$, we consider the specialised character
$\chi_\r(\tau)$ for complex numbers $\tau$ with ${\rm Im}\; \tau >0$
as follows:
\[
\chi_\r(\tau)=\Tr\!\big(e^{2\pi i\tau (L_{0,\r}-c/24)}\big).
\]
Here the operator $L_{0,\r}$ is the conformal Hamiltonian in 
the representation $\r$.  We also assume that the above trace converges, and so each eigenspace 
of $L_{0,\r}$ is finite dimensional.  

In many cases the $T$ and $S$ matrices give an action
of $SL(2,{\mathbb Z})$ on the linear span of these specialised
characters through change of variables $\tau$ as follows:
\begin{equation}\label{modularity}
\begin{split}
\chi_\r(-1/\tau)&= \sum_\nu S_{\r,\nu} \chi_\nu(\tau),\\
\chi_\r(\tau+1)&= \sum_\nu T_{\r,\nu} \chi_\nu(\tau).
\end{split}
\end{equation}
This is always the case if the Rehren matrices $S$ and $T$ coincide with the  
so called Kac-Peterson or Verlinde matrices.
The Kac-Peterson-Verlinde matrix $T$ and Rehren matrix $T$ always coincide
up to a phase by the spin-statistics theorem \cite{GL2}.

We shall say that $\B$ is \emph{modular} if  the Rheren matrices give the modular 
transformations of specialized character as in Eq. (\ref{modularity}) .  
Note that we are assuming that $\B$ is completely rational, namely the $\mu$-index is finite 
(see below) and the split property holds. However, as we are assuming
$\Tr(e^{-tL_{0}})<\infty$ the split property follows, see \cite{BDL}. 
(Strong additivity follows from  \cite{LX}.)

Modularity holds in all computed rational cases, cf. \cite{X3,X4}. 
The $SU(N)_k$ nets and the Virasoro nets ${\rm Vir}_c$ with $c<1$ are
both modular. We expect all local conformal completely rational 
nets to be modular (see \cite{Hu} for results of similar kind). 
Furthermore, if $\B
$ is a modular local conformal net and $\C
$ an irreducible extension of $\B
$, then $\C
$ is also modular \cite{KL05}, that allows to check 
the modularity property in several cases.

If $\B
$ is modular, the Kac-Wakimoto formula holds
\begin{equation}\label{KW1}
d(\r)=\frac{S_{\r,0}}{S_{0,0}}=
\lim_{\tau\to i\infty}\frac{\sum_\nu S_{\r,\nu}\chi_\nu(\tau)}
{\sum_\nu S_{0,\nu}\chi_\nu(\tau)}=
\lim_{\tau\to 0}\frac{\chi_\r(\tau)}{\chi_0(\tau)}.
\end{equation}
Here we denote the vacuum sector $\iota$ also by $0$.

Now, if $\B$ is a modular net, then $\B$ is two-dimensional 
log-elliptic with noncommutative area $a_0 = 2\pi c/24$ \cite{KL05}, 
indeed the following asymptotic formula holds:
\[
\log\Tr(e^{-2\pi tL_{0,\r}})\sim \frac{\pi c}{12}\frac1t + 
\frac{1}{2}\log\frac{d(\r)^2}{\mu_{\B
}}-
\frac{\pi c}{12}t\qquad {\rm as }\ t\to 0^+ \ .
\]
where $\r$ a representation of $\B$ and $L_{0,\r}$ is the conformal 
Hamiltonian in the representation $\r$.

\subsection{Fermi nets and modularity}
Assume that $\l$ is a graded irreducible general soliton of the Fermi conformal net $\A$ 
on $S^1$. Then $\l$ is the restriction of a representation of the double cover $\A^{(2)}$ of 
$\A$ and is diffeomorphism covariant. We denote by $H_{\l}$ the conformal Hamiltonian of 
$\A$ in the representation $\l$. 

Then $H_{\l}$ is affiliated to the Virasoro subnet in the representation $\l$, which is 
contained in the Bosonic subnet, so $H_{\l}$
commutes with $\Ga_{\l}$ and thus respects the graded decomposition 
$\H_\l = \H_{\l,+}\oplus\H_{\l,-}$ given by $\Ga_\l$; we then have a unitary equivalence
\[
H_{\l} \simeq L_{0,\r}\oplus L_{0,\r'}
\]
where $L_{0,\r}$ and $L_{0,\r'}$ are the conformal Hamiltonians of $\A_b$ 
in the  representations $\r$ and $\r'$, where $\l_b = \r\oplus\r'$. 
Consequently
\begin{equation}\label{strp}
\Str(e^{-tH_{\l}}) = \Tr(e^{-tL_{0,\r}}) - \Tr(e^{-tL_{0,\r'}})
\end{equation}
We also set
\[
\tilde H_{\l}\equiv H_{\l} - c/24,\quad\tilde L_{0,\r} \equiv 
L_{0,\r} - c/24\dots
\]
If $\A_b$ is modular we then have by formula
\begin{align}\label{m1}
\Str(e^{-2\pi t{\tilde H_{\l}}}) =&\sum_\nu
S_{\r,\nu}\Tr(e^{-2\pi\tilde L_{\r,\nu}/t}) -\sum_\nu 
S_{\r',\nu}\Tr(e^{-2\pi\tilde L_{\r',\nu}/t})\\
=& \sum_\nu
(S_{\r,\nu} - S_{\r',\nu})\Tr(e^{-2\pi\tilde L_{0,\nu}/t})\\
=& \sum_\nu
D_{\r,\nu}\Tr(e^{-2\pi\tilde L_{0,\nu}/t})
\end{align}
where $D_{\r,\nu}\equiv S_{\r,\nu} - S_{\r',\nu}$ as before and, by Prop. \ref{DS},  $D_{\r,\nu}= 2S_{\r,\nu}$ if $\nu$ is $\s$-Fermi, and zero otherwise.

\subsection{Fredholm and Jones index within superconformal invariance}
We shall say that a general representation $\l$ of the Fermi local conformal net $\A$ is 
\emph{supersymmetric}\footnote{A characteristic feature of supersymmetry is also that the graded derivation implemented by $Q_\l$ is densely defined in a appropriate sense. We shall not need this property in this paper.}
if $\l$ is graded and the conformal Hamiltonian $H_{\l}$ satisfies
\[
\tilde H_{\l} \equiv H_{\l} - c/24 = Q_{\l}^2
\]
where $Q_{\l}$ is a selfadjoint on $\H_{\l}$ which is odd w.r.t. the grading unitary $\Ga_{\l}$. 

Let then $\l$ be supersymmetric. An immediate important consequence is a positive lower bound for the energy:
\[
H_{\l}\geq c/24\ .
\]
Note that by McKean-Singer lemma $\Str(e^{-t(H_{\l} - c/24)})$ is constant thus, taking the limit as $t\to\infty$, we have
\begin{equation}\label{Windex}
\Str(e^{-t(H_{\l} - c/24)}) = \dim\ker(H_{\l} - c/24)\ ,
\end{equation}
the multiplicity of the lowest eigenvalue $c/24$ of $H_{\l}$.

Let $\r$ be one of the two irreducible components of $\l_b$ as above. We denote by $\mathfrak R$ the set of $\s$-Fermi irreducible sectors of $\A_b$, that correspond to the Ramond sectors of $\A$, see Prop. \ref{acov}. Note that $\mathfrak R \neq\emptyset$ as otherwise $\s$ would be a degenerate sector, which is not possible because the braided tensor category of DHR 
sectors in modular as a consequence of the complete rationality assumption.

Assume now that $\A_b$ is modular. We notice that formula \eqref{m1} can be written as
\begin{equation}\label{expmin}
\Str(e^{-2\pi t\tilde H_{\l}}) = 2\sum_{\nu\in\mathfrak R}
S_{\r,\nu}\Tr(e^{-2\pi\tilde L_{0,\nu}/t}) \ .
\end{equation}
\begin{corollary} We have
\[
\sum_{\nu\in\mathfrak R } S_{\r,\nu}d(\nu) = 0
\]
\end{corollary}
\proof
Divide by $\Tr(e^{-2\pi tL_0})$ both members of (\ref{expmin}); 
the statement then follows by the Kac-Wakimoto formula \eqref{KW1}.
\endproof
\begin{sublemma}
Let $A_{1}, A_{2},\dots A_{n}$ be selfadjoint operators such that $\Tr(e^{-sA_k})<\infty$  and $\sum_{k}c_{k}\Tr(e^{-sA_k})$ is a constant function of $s>0$, for some scalars $c_{k}$. Then 
\[
\lim_{s\to +\infty}\sum_{k}c_{k}\Tr(e^{-sA_k})=\sum_{k}c_{k}\Dim\ker A_k
\]
\end{sublemma}
\proof
Clearly $\lim_{s\to +\infty}\Tr(e^{-sA})=\Dim\ker A$ if $A$ is trace class and non-negative. The finite trace condition implies that $A_k$ bounded below.
Let $A_k = A^{+}_k + A^-_k$ with $A^{+}_k\geq 0$ and $A^{-}_k <0$. Then the function $\sum_{k}c_{k}\Tr(e^{-sA^{-}_k})$ is a linear combination of exponentials of the form $e^{as}$ with $a>0$, and vanishes at infinity, thus it must be identically zero. It follows that
\[
\sum_{k}c_{k}\Tr(e^{-sA_k})=\sum_{k}c_{k}\Tr(e^{-sA^+_k})\underset{s=\infty}{\longrightarrow}
\sum_{k}c_{k}\Dim\ker A^+_k = \sum_{k}c_{k}\Dim\ker A_k
\]
\endproof
\begin{lemma}\label{ex1}
Let $\A$ be a Fermi modular net. If $\l$ is supersymmetric then 
\begin{equation}\label{lim}
\Str(e^{-2\pi t\tilde H_\l}) = 2\sum_{\nu\in\mathfrak R} S_{\r,\nu}\fn(\nu,c/24)
\end{equation}
where $\fn(\nu,h)\equiv\Dim\ker(L_{0,\nu}- h)$.
\end{lemma}
\proof
The left hand side of  \eqref{expmin} is constant by McKean-Singer lemma, so we have that
$\sum_{\nu\in\mathfrak R}S_{\r,\nu}\Tr(e^{-2\pi s\tilde L_{0,\nu}})$ is a constant function of 
$s>0$. 
Thus  \eqref{lim} holds by the sublemma.
\endproof
With $\l$ as in the above lemma, by eq. \eqref{Windex} we have
\[
\Str(e^{-2\pi t\tilde H_{\l}}) = \ind(Q_{\l +}
)\ .
\]
Therefore by Lemma \ref{ex1} we have
\[
\ind(Q_{\l +}
) = 2\sum_{\nu\in\mathfrak R}  S_{\r,\nu}\fn(\nu,c/24)
\]
then, writing Rehren definition of the $S$ matrix, we have
\[
\ind(Q_{\l +}
) = 2\frac{d(\r)}{\sqrt{\mu_{\A_b}}}\sum_{\nu\in\mathfrak R}K(\r,\nu)
d(\nu)\fn(\nu,c/24)
\]
where $\mu_{\A_b}$ is the $\mu$-index of $\A_b$. By \cite{KLM} we have $\mu_{\A_b} = 4\mu_{\A}$  therefore:
\begin{theorem}\label{FJ}
Let $\A$ be a Fermi conformal net as above and $\l$ 
a supersymmetric irreducible representation of $\A$. Then
\[
\ind(Q_{\l +}
) = 
\frac{d(\r)}{\sqrt{\mu_\A}}\sum_{\nu\in\mathfrak R}K(\r,\nu)
d(\nu)\fn(\nu,c/24)
\]
where $\r$ is one of the two irreducible components of $\l_b$ and $\mu_\A$ is the $\mu$-index of $\A$.
\end{theorem}
\noindent
In the above formula the Fredholm index of the supercharge operator $Q_{\l +}
$ is expressed 
by a formula involving the Jones index of the Ramond representations whose lowest eigenvalue $c/24$ modulo integers.
\begin{corollary} If $\ind(Q_{\l})\neq 0$ there exists a $\s$-Fermi sector $\nu$ such that  $c/24$  is an eigenvalue of $L_{0,\nu}$.
\end{corollary}
\begin{corollary} Suppose that, in Th. \ref{FJ}, $\r$ is the only Ramond sector with lowest eigenvalue $c/24$ modulo integers. Then 
\[
S_{\r,\r} = \frac{d(\r)^2}{\sqrt{\mu_{\A_b}}}K(\r,\r) = \frac12 \ .
\]
\end{corollary}
\proof
By Lemma \ref{ex1} we have:
\[
\ind(Q_{\l +}
) = 2S_{\r,\r}\fn(\r,c/24)\ .
\]
On the other hand by formula \eqref{Windex} 
\[
\ind(Q_{\l +}
) = \dim\, {\rm ker}(H_\l - c/24) = \fn(\r,c/24) + \fn(\r',c/24)=  \fn(\r,c/24) 
\]
because $H_\l = L_{0,\r}\oplus L_{0,\r'}$ and $c/24$ is not in the spectrum of $L_{0,\r'}$, so we get our formula.
\endproof
\section{Super-Virasoro algebra and super-Virasoro nets}
We now focus on model analysis and shall consider the most basic superconformal nets, namely the ones associated with the super-Virasoro algebra.
\subsection{Super-Virasoro algebra}
The super-Virasoro algebra governs the superconformal invariance \cite{FQS,GKO}. It plays in the supersymmetric context the same r\^ole that the Virasoro algebra  plays in the local conformal case.

Strictly speaking, there are two super-Virasoro algebras. They are the super-Lie algebras generated by even 
elements $L_n$, $n\in\mathbb Z$, odd elements $G_r$, and a central even element $c$, 
satisfying the relations
\begin{gather}\label{svirdef}
    [L_m , L_n] = (m-n)L_{m+n} + \frac{c}{12}(m^3 - m)\de_{m+n, 0}\\
    \nonumber
    [L_m, G_r] = (\frac{m}{2} - r)G_{m+r}\\
    \nonumber
    [G_r, G_s] = 2L_{r+s} + \frac{c}{3}(r^2 - \frac14)\de_{r+s,0}
    \end{gather}
Here the brackets denote the super-commutator. 
In the \emph{Neveu-Schwarz} case $r\in\mathbb Z + 1/2$, while in the 
\emph{Ramond} case $r\in\mathbb Z$. We shall sometime use the term 
\emph{super-Virasoro algebra} to indicate either the Neveu-Schwarz 
algebra or the Ramond algebra. 

The point is that, although the Neveu-Schwarz algebra and the  Ramond algebra are not isomorphic graded Lie algebras, they are representations of a same object, in a sense that we shall later see (they have the same ``isomorphic completion").

By definition, the Neveu-Schwarz algebra and the Ramond algebra are both
extensions of the Virasoro algebra. 

The super-Virasoro algebra is equipped with the involution $L^*_n = 
L_{-n}$, $G^*_r = G_{-r}$, $c^* = c$ and we will be only interested in unitary 
representations on a Hilbert space,
i.e. representations preserving the involution. Note that unitary representations have automatically positive energy, namely $L_0\geq 0$. Indeed we have 
\begin{gather*}
  L_0 =  \frac12[G_{\frac12}, G_{-\frac12}] = \frac12\left(G_{\frac12}G_{\frac12}^* +
  G_{\frac12}^*G_{\frac12}\right)\geq 0\qquad (\text{NS case})\\
  L_0 =  \frac12[G_0, G_0] = G_0^2 + c/24 \geq c/24
\qquad (\text{R case})    
\end{gather*}
The unitary, lowest weight representations of the super-Virasoro 
algebra, namely the unitary representations of the super-Virasoro 
algebra on a Hilbert space $\H$ with a cyclic vector 
$\xi\in \H$ satisfying
\[
L_0\xi = h\xi, \quad L_n\xi = 0, \ n>0, \quad G_r\xi = 0,\ r>0,
\]
are studied in \cite{FQS,GKO}; in the NS case they are irreducible and uniquely determined 
by the values of $c$ and $h$. In the Ramond case one has to further specify the action of $G_0$ on the lowest energy subspace. It turns out that for a possible value of $c$ and $h$ there two inequivalent irreducible lowest weight representations (but for the case $c = h/24$ when the representation is unique and graded). Not also that for $c \neq h/24$ the direct sum of the two inequivalent irreducible representations has a cyclic lowest energy vector and it is the unique Ramond $(c,h)$ lowest weight representation where grading is implemented.

The possible values are either $c\geq 3/2$, $h\geq 0$ or
\begin{equation}\label{c-values}
c = \frac32 \left( 1 - \frac{8}{m(m+2)}\right), \ m=2,3,\ldots
\end{equation}
and
\[
h= h_{p,q}(c)\equiv\frac{[(m+2)p - mq]^2 - 4}{8m(m+2)} + \frac{\e}{8}
\]
where $p= 1,2,\ldots,m-1$, $q=1,2,\ldots, m+1$ and $p-q$ is even or 
odd corresponding to the Neveu-Schwarz case ($\e= 0$) or Ramond 
case ($\e = 1/2$).

Note that the Neveu-Schwarz algebra has a vacuum representation, 
namely a irreducible representation with $0$ as eigenvalue of $L_0$, the Ramond 
algebra has no vacuum representation. 
\subsection{Stress-energy tensor}
\label{set}
Let $c$ be an admissible value as above with $L_n$ $(n\in\mathbb Z)$, 
$n\in\mathbb Z$, $G_r$, the operators in 
the corresponding to a Neveu-Schwarz ($r\in\mathbb Z + \frac12$) or Ramond ($r\in\mathbb Z$)
representation. The Bose and Fermi stress-energy 
tensors are defined by
\begin{align}
T_B(z) =& \sum_{n} z^{-n-2} L_n\\
T_F(z) =& \frac12 \sum_r z^{-r-3/2} G_r
\end{align}
namely
\begin{equation}
  \text{Neveu-Schwarz case:}\begin{cases}
T_{B}(z) =& \sum_{n\in\mathbb Z} z^{-n-2} L_n\\
T_{F}(z) =& \frac12 \sum_{m\in\mathbb Z} z^{-m-2} G_{m+\frac12}
\end{cases}
\end{equation}
\begin{equation}  \text{Ramond case:}\begin{cases}
T_{B}(z) =& \sum_{n} z^{-n-2} L_n\\
T_{F}(z) =& \frac12 \sum_{m\in\mathbb Z} z^{-m-2} \sqrt{z} G_m
\end{cases}
\end{equation}
Let's now make a formal calculation for the (anti-)commutation relations of the Fermi stress energy tensor $T_F$. We want to show that, setting $w\equiv z_2 /z_1$, we have
\begin{equation}\label{comrel}
[T_F(z_1), T_F(z_2)]= \frac12 z_1^{-1}T_B(z_1) \de(w)
+ z_1^{-3}w^{-\frac32}\frac{c}{12}
\big(w^2\de''(w)  + \frac34\de(w)\big)
\end{equation}
both in the Neveu-Schwarz and in the Ramond case.

Setting $k= r=s\in\mathbb Z$ we have, say in the Ramond case, 
\begin{align*}
[T_F&(z_1), T_F(z_2)]= \frac14\sum_{r,s}[G_r, G_s]z_1^{-r-3/2}z_2^{-s-3/2}\\ 
&= \frac12\sum_{r,s}L_{r+s}z_1^{-r-3/2}z_2^{-s-3/2} 
+  \sum_r \frac{c}{12}\big(r^2 - \frac14\big)z_1^{-r-3/2}z_2^{r-3/2}\\
&=\frac12\sum_{r,k}L_{k}z_1^{-r-3/2}z_2^{-k+r-3/2}
 + z_1^{-3/2}z_2^{-3/2}\sum_r\frac{c}{12}\big(r^2 - \frac14\big)w^r\\
&= \frac12 z_1^{-1}\Big(\sum_k L_{k}z_2^{-k-2}\Big) w^{\frac12}\sum_r w^{r} + z_1^{-3}w^{-\frac32}\sum_r\frac{c}{12}\big(r^2 - \frac14\big)w^r\\
&=\frac12 z_1^{-1}T_B(z_2) w^{\frac12}
\sum_r w^r + z_1^{-3} w^{-\frac32}\sum_r\frac{c}{12}\big(r^2 - \frac14\big)w^r\\
&= \frac12 z_1^{-1}T_B(z_2) w^{\frac12}\de(w)
+ z_1^{-3}w^{-\frac32}\frac{c}{12}
\Big(w^2\de''(w)  + \big(w-\frac14\big)\de(w)\Big)\\
&= \frac12 z_1^{-1}T_B(z_1) \de(w)
+ z_1^{-3}w^{-\frac32}\frac{c}{12}
\Big(w^2\de''(w)  + \frac34\de(w)\Big)
\end{align*}
As
\[
\sum_{r\in\mathbb Z} \big(r^2 - \frac14\big)w^r 
=\sum_{r\in\mathbb Z +\frac12} \big(r^2 - \frac14\big)w^r
= w^2\de''(w)  + \frac34\de(w)
\]
the above calculation (by using equalities as $\de(w)\sqrt{w}=\de(w)$ and similar ones) shows that the commutation relations for the Fermi stress energy tensor $T_F$ and, analogously, the commutation relations for $T_B$ and $T_F$ are indeed the same in the Neveu-Schwarz case and in the Ramond case, namely they are representations of the same (anti-)commutation relations. This is basic reason to view the Neveu-Schwarz and Ramond algebras as different types of representations of a unique algebra. 

However the above calculation is only formal. To give it a rigorous meaning, 
and have convergent series, we have to smear the stress energy tensor with a 
smooth test function with support in an interval. We then arrive naturally to 
consider the net of von Neumann algebras of operators localised in intervals. 
In the case of central charge $c<3/2$ we shall see that Neveu-Schwarz and 
Ramond representations correspond to DHR and general 
solitons of the associated super-Virasoro net.

\subsection{Super-Virasoro nets}
We give here the definition of the super-Virasoro nets for all the 
allowed values of the central charge. 
We follow the strategy adopted in \cite{BS-M} for the case of Virasoro nets, 
cf. also \cite{Car04,loke}. An alternative construction in the case of the 
discrete series ($c<3/2)$ is outlined in Sect. \ref{SVirnet}. 

Let $\l$ be a unitary positive energy representation of the 
super-Virasoro algebra on a Hilbert space ${\mathcal H}_\l$. The 
corresponding conformal Hamiltonian $L_0$ is the self-adjoint operator on $\H_\l$ 
and we denote $\H_\l^\infty$ the dense subspace of smooth vectors for $L_0$, 
namely the subspace of vectors belonging to the domain of $L_0^n$ for all 
positive integers $n$.

The operators $L_n$,
$n\in {\mathbb Z}$ satisfy the linear energy bounds 
\begin{equation}
\label{e-boundsB}
\|L_n v\|\leq M (1+|n|^{\frac{3}{2}})\|(1+L_0)v\|, \quad v\in {\mathcal 
H}_\l^\infty, 
\end{equation}
for suitable constant $M>0$ depending on the central charge $c$, cf. 
\cite{BS-M,CW}. 
Moreover from the relations 
$[G_{-r}, G_r] = 2L_{0} + \frac{c}{3}(r^2 - \frac14)$, 
we find the energy bounds
\begin{equation}
\label{e-boundsF}
\|G_r v\|\leq (2+ \frac{c}{3}r^2)^{\frac12}\|(1+L_0)^{\frac12}v\|, \quad v\in 
{\mathcal H}_\l^\infty,
\end{equation}
where $r \in \ZZ+1/2$ (resp. $r \in \ZZ$) if $\l$ is a Neveu-Schwarz 
(resp. Ramond) representation. 
We now consider the vacuum representation of the super-Virasoro algebra with 
central charge $c$ and denote by $\H$ the corresponding Hilbert space and 
by $\Omega$ the vacuum vector, namely the unique (up to a phase) unit 
vector such that $L_0\Omega=0$.  

Let $f$ be a smooth function on $S^1$. It follows from the linear 
energy bounds in Eq. (\ref{e-boundsB}) and the fact that the Fourier 
coefficients
\begin{equation}
\hat{f}_n=\int_{-\pi}^\pi f(e^{i\theta})e^{-in\theta}\frac{{\rm d}\theta}{2\pi},\;
n\in \ZZ,
\end{equation}
are rapidly decreasing, that the smeared Bose stress-energy tensor  
\begin{equation}
T_B(f)= \sum_{n \in \ZZ}\hat{f}_nL_n
\end{equation}
is a well defined operator with invariant domain $\H^\infty$. Moreover, for 
$f$ real,  $T_B(f)$ is essentially self-adjoint on $\H^\infty$ (cf. \cite{BS-M})
and we shall denote again $T_B(f)$ its self-adjoint closure.  

Now let $f$ be a smooth function on $S^1$ whose support do not contains 
$-1$. Then also the coefficients  
\begin{equation}
\hat{f}_r=\int_{-\pi}^\pi f(e^{i\theta})e^{-ir\theta}\frac{{\rm d}\theta}{2\pi},\;
r\in \ZZ + \frac{1}{2},
\end{equation}
are rapidly decreasing and it follows from the energy bounds 
in Eq. (\ref{e-boundsF}) that the corresponding smeared Fermi stress-energy 
tensor 
\begin{equation}
T_F(f)=\frac{1}{2}\sum_{r \in \ZZ+\frac{1}{2}}\hat{f}_rG_r
\end{equation}
is also a well defined operator with invariant domain $\H^\infty$. Again, 
for $f$ real, $T_F(f)$ is essentially self-adjoint on $\H^\infty$ 
(cf. \cite{BS-M}) and we denote its self-adjoint closure by the same 
symbol. 

As in Sect. \ref{netR} we identify $\mathbb R$ with $S^1/ \{-1\}$ and
consider the family $\I_{\mathbb R}$ of nonempty, bounded, open intervals of
${\mathbb R}$ as a subset of the family $\I$ of intervals of $S^1$. We define
a net ${\rm SVir}_{c}$ of von Neumann algebras on $\RR$ by 
\begin{equation}
\label{gener1} 
{\rm SVir}_{c}(I) \equiv \{ e^{iT_B(f)}, e^{iT_F(f)}:f\in C^{\infty}(S^1) 
\;{\rm real},\, {\rm supp}f, \subset I\}'', \; I \in \I_\RR. 
\end{equation} 
Isotony is clear from the definition and we have to show that the net is graded 
local and covariant.  

We first consider graded locality. The spectrum of $L_0$ is 
contained in $\ZZ/2$ and hence the unitary operator $\Gamma = e^{i2\pi L_0}$ 
is an involution such that $\Gamma \Omega$.  
It is straightforward to check that if $f_1, f_2$ 
are smooth functions on $S^1$ and the support of $f_2$ does not contain $-1$ 
then $\Gamma T_B(f_1) = T_B(f_1)\Gamma$ and 
$\Gamma T_F(f_1) = -T_F(f_1) \Gamma$.  
Hence, $\gamma = {\rm Ad}\Gamma$ is a $\ZZ_2$-grading on the net 
${\rm SVir}_{c}$, namely 
$\Gamma {\rm SVir}_{c}(I) \Gamma^* ={\rm SVir}_{c}(I)$, for all $I\in 
\I_\RR$. Now let $I_1, I_2 \in \I_\RR$ be disjoint intervals and 
let $f_1, f_2$ be real smooth functions on $S^1$ with support in $I_1, 
I_2$ respectively. Then the operators $T_B(f_1)$ and $T_F(f_1)$ commute
with $ZT_B(f_2)Z^*$ and $ZT_F(f_2)Z^*$, $Z=\frac{1-i\Gamma}{1-i}$, on 
$\H^\infty$, cf. the (anti-)commutation relations in Sect. \ref{set}
(note that $Z\H^\infty=\H^\infty)$. Using the energy  bounds in 
Eq. (\ref{e-boundsB}) and Eq. (\ref{e-boundsF}) and the fact that 
$Z$ commutes with $L_0$ one can apply the argument in 
\cite[Sect. 2]{BS-M} to show that $e^{iT_B(f_1)}$ and $e^{iT_F(f_1)}$ 
commute with $Ze^{iT_B(f_2)}Z^*$ and $Ze^{iT_F(f_2)}Z^*$. It follows that
the net ${\rm SVir}_c$ is graded local, namely 
\begin{equation} 
{\rm SVir}_c(I_1) \subset Z{\rm SVir}_c(I_2)Z^*
\end{equation}  
whenever $I_1, I_2$ are disjoint interval in $\I_\RR$. 

We now discuss covariance. The crucial fact here is that the representation 
of the Virasoro algebra on $\H$ integrates to a strongly continuous 
unitary projective positive-energy representation of $\Diff^{(\infty)}(S^1)$ 
on $\H$ by \cite{GoWa,Tol99} which factors through $\Diff^{(2)}(S^1)$ because
$e^{i4\pi L_0}=1$. Hence there is a strongly continuous projective 
unitary representation $U$ of $\Diff^{(2)}(S^1)$ on $\H$ such that, 
for all real $f \in C^\infty(S^1)$ and all $x \in B(\H)$,  
\begin{equation}
U(\exp^{(2)}(tf))xU(\exp^{(2)}(tf))^*= 
e^{itT_B(f)}xe^{-itT_B(f)}, 
\end{equation}
where, $\exp^{(2)}(tf)$ denotes the lift to $\Diff^{(2)}(S^1)$ of the
one-parameter subgroup $\exp(tf)$ of $\Diff(S^1)$ generated by the 
(real) smooth vector field $f(e^{i\theta})\frac{{\rm d}}{{\rm d}\theta}$. Moreover, 
if $\theta \to r^{(2)}(\theta)$ is the lift to $\Diff^{(2)}(S^1)$ of the 
one-parameter subgroup of rotations in $\Diff(S^1)$ we have  
\begin{equation}
U(r^{(2)}(\theta))= e^{i\theta L_0},
\end{equation}
for all $\theta$ in $\RR$. 
The following properties of $U$ follow rather straightforwardly. 
\medskip

\noindent $(1)$ The restriction of $U$ to the subgroup 
$\Mob^{(2)} \subset \Diff^{(2)}(S^1)$ lifts to a unique strongly continuous 
unitary representation which we again denote by $U$. 
If $\exp^{(2)}(tf) \in \Mob^{(2)}$ for all $t\in \RR$, this 
unitary representation satisfies 
\begin{equation}
U(\exp^{(2)}(tf))=e^{itT_B(f)},\; U(\exp^{(2)}(tf))\Omega=\Omega,
\end{equation}
for all $t\in \RR$.   

\medskip

\noindent $(2)$  If the support of the real smooth function $f$ is contained 
in $I \in \I_\RR$ then 
\begin{equation}
\label{covTB}
U(g)e^{iT_B(f)}U(g)^*\in {\rm SVir}_c(\dot{g}I)
\end{equation}
for all $g \in \Diff^{(2)}(S^1)$ such that $\dot{g}I \in \I_\RR$. 
\medskip

\noindent $(3)$ For all $g \in \Diff^{(2)}(S^1)$ we have 
\begin{equation}
U(g)\Gamma U(g)^* = \Gamma
\end{equation}
for all $g \in \Diff^{(2)}(S^1)$.

\medskip

Note that property $(3)$ follows from the fact that $\Diff^{(2)}(S^1)$ is 
connected and $\Gamma^2=1$. 

We now consider the covariance properties of the Fermi stress-energy
tensor $T_F$. From the commutation relations in equations (\ref{svirdef})
we find,  
\begin{equation}
\label{smearedcommutator}
i[T_B(f_1), T_F(f_2)]v = T_F(\frac{1}{2}f'_1f_2-f'_2f_1)v, \; v\in 
\H^\infty
\end{equation}
where $f_1, f_2$ are real smoot functions such that 
${\rm supp} f_2 \subset I$ for some $I\in \I_\RR$ and, for any
$f\in C^{\infty}(S^1)$, 
$f'$ is defined by $f'(e^{i\theta})=\frac{d}{d\theta}f(e^{\theta})$.  
For any $g \in \Diff(S^1)$ consider the function 
$X_g:S_1 \to \RR$ defined by
\begin{equation}
X_g(e^{i\theta})= -i\frac{{\rm d}}{{\rm d}\theta}\log (ge^{i\theta}).
\end{equation}
Since $g$ is a diffeomorfism of $S^1$ preserving the orientation 
then $X_g(z)>0$ for all $z \in S^1$. Moreover $X_g \in C^\infty(S^1)$. 
Another straightforward consequence of the definition is that 
\begin{equation}
X_{g_1g_2}(z)=X_{g_1}(g_2z)X_{g_2}(z). 
\end{equation}
As a consequence the the family of continuous linear operators 
$\beta(g)$, $g\in \Diff (S^1)$  on the Fr\'echet space $C^\infty(S^1)$ 
defined by 
\begin{equation}
(\beta(g)f)(z)= X_g(g^{-1}z)^{\frac12}f(g^{-1}z)
\end{equation}
gives a strongly continuous representation of $\Diff (S^1)$ leaving 
the real subspace of real functions invariant. Moreover if 
$f_1, f_2\in C^\infty(S^1)$ are real then vector valued function 
$t \to \beta(\exp(tf_1))f_2$ is differentiable in $C^\infty(S^1)$ and
\begin{equation}
\label{actionderivative}
\frac{d}{dt}\beta(\exp(tf_1))f_2|_{t=0}=\frac{1}{2}f'_1f_2-f_1f'_2.
\end{equation}  
Now let ${\rm supp} f_2$ be a subset of some interval $I \in I_\RR$ 
and let $J_I\subset \RR$ be the connected component of $0$ in $\RR$ 
of the open set $\{t\in \RR:\exp(tf_1)I \in I_\RR \}$. Then, for any 
$v\in \H^\infty$ the function $J_I \ni t \to T_F(\beta(\exp(tf_1))f_2)v$ 
is differentiable in $\H$ and it follows from Eq. 
(\ref{smearedcommutator}) and Eq. (\ref{actionderivative}) 
that 
\begin{equation}
\label{dTF} 
\frac{d}{dt}T_F(\beta(\exp(tf_1))f_2)v|_{t=0}=i[T_B(f_1), T_F(f_2)]v.
\end{equation}
We now specialize to the case of M\"{o}bius transformations i.e. 
we assume that $\exp(tf_1) \in \Mob$ for all $t\in \RR$.   
The map $J_I \ni t \to v(t) \in \H$ given by 
\begin{equation}
v(t)=T_F(\beta(\exp(tf_1))f_2)U(\exp^{(2)}(tf_1))v
\end{equation} 
is well defined because, $U(\exp^{(2)}(-tf_1))\H^\infty = \H^\infty$ 
for all $t\in \RR$. Note also that $v(t)\in \H^\infty$ for all $t\in J_I$.
Now using Eq. (\ref{dTF}) and the energy bounds in Eq. 
(\ref{e-boundsF}) it can be shown that $v(t)$ is differentiable (in the 
strong topology of $\H$) and that it satisfy the following differential 
equation on $\H$ 
\begin{equation}
\frac{d}{dt}v(t)=iT_B(f_1)v(t).
\end{equation}
If follows that 
\begin{equation}
v(t)=U(\exp^{(2)}(tf_1))T_F(f_2)v
\end{equation}
and since $v \in \H^\infty$ was arbitrary we get, for all $t\in J_I$, 
the following equality of self-adjoint operators
\begin{equation}
U(\exp^{(2)}(tf_1))T_F(f_2)U(\exp^{(2)}(-tf_1))=T_F(\beta(\exp(tf_1))f_2). 
\end{equation}
Now, if we denote by $\U^{(2)}_I$ the connected component of the identity 
in $\Mob^{(2)}$
of the open set $\{g\in \Mob^{(2)}: gI \in \I_\RR \}$ it follows that 
\begin{equation}
\label{covTF}
U(g)T_F(f)U(g)^*=T_F(\beta(\dot{g})f),
\end{equation}
for any real smoot function on $S^1$ with ${\rm supp}f \subset I$
and all $g\in \U^{(2)}_I$. 

From Eq. (\ref{covTF}) and Eq. (\ref{covTB}) we have 
\begin{equation}
U(g){\rm SVir}_c (I)U(g)^* = {\rm SVir}_c(\dot{g}I)\; I\in \I_\RR,
\; g \in \U^{(2)}_I. 
\end{equation} 
Hence ${\rm SVir}_c$ extends to a M\"{o}bius covariant  net on 
$S^1$ satisfying graded locality, see Sect. \ref{coverNets}. Note that we 
have not yet shown that the vacuum vector $\Omega$ is cyclic and 
hence we still don't knew if the net satisfy all the requirements of 
Property 3 in the definition of M\"{o}bius covariant Fermi nets on
$S^1$ given in Sect. \ref{MobFermiNets}. 
We shall however prove the cyclicity of the vacuum as a part of the following 
theorem. 

\begin{theorem}    
${\rm SVir}_c$ is an irreducible Fermi conformal on $S^1$ for any of the 
allowed values 
of the central charge $c$.
\end{theorem}
\proof Since $\Omega$ is the unique (up to a phase) unit vector in 
the kernel of $L_0$, we only have to show that $\Omega$ is cyclic and that 
the strongly continuous positive-energy projective representation $U$
of $\Diff^{(2)}(S^1)$ defined above makes the net diffeomorphism 
covariant in the appropriate sense. We first show that $\Omega$ 
is cyclic for the net. Let $\K \subset \H$ be the closure 
of $\bigvee_{I\in \I} {\rm SVir}_c(I) \Omega$. 
We have to show that $\K =\H$.
Clearly $U(g)\K=\K$
for all $g \in \Mob^{(2)}$. It follows that if $j\in \ZZ /2$ 
$P_j$ is the orthogonal projection of $\H$ onto the kernel of $L_0-k1$ 
then $P_j \K \subset \K\cap \H^\infty$. Now let $r\in \ZZ+1/2$. Since the 
smooth functions on $S^1$ wose support does not contain the point $-1$ is 
dense in $L^2(S^1)$ we can find an interval $I\in \I_\RR$ and 
a real smooth function $f$ with ${\rm supp}f \subset I$ such that 
$\hat{f}_r \neq 0$. Since $T_F(f)P_j \K \subset \K$ we find 
\begin{equation}
G_rP_j\K=\frac{1}{\hat{f}_r }P_{j-r}T_F(f)P_j\K \subset \K. 
\end{equation}
A similar argument applies to the operators $L_n$, $n\in \ZZ$ 
and hence the linear span $\L$ of the subspaces 
$P_j\K$, $j\in \ZZ /2$ is invariant for the 
representation of the super-Virasoro algebra. Since $\Omega \in \L$ 
is cyclic for the latter representation it follows that $\L$ is 
dense in $\H$. Hence $\K = \H$ because $\L \subset \K$.  Hence 
${\rm SVir}_c$ is an irreducible M\"{o}bius covariant Fermi net 
on ${S^1}$.  

To show that ${\rm SVir}_c$ is diffeomorphism covariant we first observe 
that by \cite[Sect. V.2]{loke}  for any $I\in \I$ the group generated by 
diffeomorphisms of the form 
$\exp(f)$ with ${\rm supp}f \subset I$ is dense in $\Diff_I(S^1)$. It 
follows that the group generated by elements of the form 
$\exp^{(2)}(f)$ with ${\rm supp}f \subset I$ is dense in 
$\Diff^{(2)}_I(S^1)$. Hence, for any $I \in \I$ and any 
$g\in \in \Diff^{(2)}_I(S^1)$,  $U(g) \in {\rm SVir}_c(I)$ and, 
by graded locality, $U(g) \in {\rm SVir}(I')'$ because. Now, an 
adaptation of the argument in the proof of \cite[Proposition 3.7]{Car04}
shows that ${\rm SVir}_c$ is diffeomorphism covariant and the proof is 
complete. 
\endproof
\subsection{The discrete series of super-Virasoro nets}\label{SVirnet}
We shall now use the construction in \cite{GKO} to study ${\rm SVir}_{c}$ with $c < 3/2$ an admissible value.
First consider three real free Fermi fields in the NS representation. 
They define a graded-local net on $S^1$.  
This net coincides with $\F^{\hat{\otimes}3}=\F\hat{\otimes}\F\hat{\otimes}\F$ 
where $\F$ is the net generated by a single real free Fermi field in the NS representation 
(cf. \cite{Bock}) and $\hat{\otimes}$ denotes the graded tensor product. The net 
$\A_{\su2_2}$ embeds as a subnet of $\F^{\hat{\otimes}3}$. Actually, from the discussion in 
\cite[page 115]{GKO} we have
\[
\A_{\su2_2}=\F^{\hat{\otimes}3}_b\ .
\] 
Now consider the conformal net $\F_N$ (on the Hilbert space 
$\H_N$) given by $\F^{\hat{\otimes}3}\otimes \A_{\su2_N}$, $N$ positive integer.  

Consider now the the representation of the 
super-Virasoro algebra on $\H_N$ with central 
charge 
\begin{equation}
c_N=\frac{3}{2}\left(1-\frac{8}{(N+2)(N+4)}\right),
\end{equation} 
constructed in \cite[Sect. 3]{GKO} (coset construction). Then the corresponding 
stress energy-tensors $T_B$ and $T_F$ generate a family of von Neumann algebras on $\H_N$ 
as in eq. \eqref{gener1}. Using the energy bounds in Eq. (\ref{e-boundsB}) and Eq. 
(\ref{e-boundsF}) it can be shown that this family defines 
a Fermi subnet of $\F_N$ as in eq. \eqref{gener1} which can be identified with 
the super-Virasoro net ${\rm SVir}_{c_N}$.
In this way we obtain all the super-Virasoro nets corresponding to the discrete series.

Using \cite{GKO} we can identify these super-Virasoro nets as 
coset subnets. From the embedding  
\begin{equation}
\A_{\su2_2}\otimes \A_{\su2_N} \subset \F_N
\end{equation}
we have the embedding
\begin{equation}
\A_{\su2_{N+2}} \subset \F_N.
\end{equation}
It follows from Eq. (3.13) and the claim at the end of 
page 114 in \cite[Sect. 3]{GKO} that ${\rm SVir}_{c_N}$ is 
contained in the coset 
\begin{equation}
(\A_{\su2_{N+2}})^c = (\A_{\su2_{N+2}})'\cap \F_N.
\end{equation}
Moreover it follows from the branching rules in 
\cite[Eq. 4.15]{GKO} that these nets coincide (cf. \cite{KL1})
namely 
\begin{equation} 
{\rm SVir}_{c_N} = (\A_{\su2_{N+2}})^c. 
\end{equation}
As a consequence the Bose subnet ${\rm SVir}^0_{c_N}\equiv ({\rm SVir}_{c_N})_b$ of 
the super-Virasoro net ${\rm SVir}_{c_N}$ is equal to the coset

\begin{equation}
\label{BoseCoset}
(\A_{\su2_{N+2}})'\cap \left( \A_{\su2_2}\otimes \A_{\su2_N} \right)
\end{equation}
and hence, by \cite[Corollary 3.4]{X2} and \cite[Theorem 24]{L03} 
${\rm SVir}^0_{c_N}
$ is completely rational, see also \cite[Corollary 28]{L03}. 

Now we look at representations. We denote $(NS)$ and $(R)$ the 
Neveu-Schwartz and Ramond representations for three Fermion fields 
respectively. In $(NS)$ the lowest energy eigenspace is one-dimensional
(``nondegenerate vacuum ''), whilst 
in $(R)$ it is two-dimensional (``2-fold degenerate vacuum'').\footnote{Different Ramond 
representations could be defined corresponding to different choices of the corresponding representation of the Dirac algebra of the 0-modes on the subspace of  lowest energy vectors, 
cf. page 113 and page 115 of \cite{GKO}.} 

It is almost obvious that $(NS)$ corresponds to the vacuum 
representation $\pi_{NS}$ of $\F^{\hat{\otimes}3}$ and, 
arguing as in the proof of \cite[Lemma 4.3]{Bock}, it can be 
shown that $(R)$ corresponds to a general soliton  
$\pi_R$ of the latter net. 
Clearly $(NS)$ and $(R)$ restrict to positive-energy 
representations of $\su2_2$. We denote by 
$\pi_{(N,l)}$ the representation of $\A_{\su2_N}$ with spin $l$. 
At level $N$ the possible values of the spin are those satisfying 
$0\leq 2l \leq N$. Then the following identities hold (see \cite[page 
116]{GKO}): 

\begin{eqnarray}
\label{NSrest} 
\pi_{NS}|_{\A_{\su2_2}} = \pi_{(2,0)} \oplus \pi_{(2,1)} \\
\label{Rrest}
\pi_R|_{\A_{\su2_2}} = \pi_{(2,\frac{1}{2})}.
\end{eqnarray}   
Note that the restriction of $(R)$ remains irreducible because the 
grading automorphism is not unitarily implemented, cf. Prop. \ref{n1}. 

Denote by $(c_N,h_{p,q})_{NS}$, resp. $(c_N,h_{p,q})_{R}$, a 
NS, resp. R, irreducible representation of super-Virasoro algebra with central 
charge $c_N$ and lowest energy 
$$h_{p,q}=\frac{\left[(N+4)p-(N+2)q\right]^2 -4}{8(N+2)(N+4)},$$
resp. 
$$h_{p,q}=\frac{\left[(N+4)p-(N+2)q\right]^2 -4}{8(N+2)(N+4)} + 
\frac{1}{16},$$
where $p=1,2,\dots,N+1$, $q=1,2, \dots N+3$ and $p-q$ is even in the 
NS case and odd in the R case. 

As already mentioned, in the NS case, for every value of the central charge, the lowest energy
$h_{p,q}$ completely determines the (equivalence class of) the 
representation. In contrast for a given values of the central charge 
and of the lowest energy $h_{p,q}$ there are two Ramond representations 
one with 
\[
G_0\Psi_{h_{p,q}}= \sqrt{h_{p,q} - \frac{c_N}{24}}\Psi_{h_{p,q}}
\]
and the other with
\[
G_0\Psi_{h_{p,q}}= -\sqrt{h_{p,q} - \frac{c_N}{24}}\Psi_{h_{p,q}},
\]
where $\Psi_{h_{p,q}}$ is the lowest energy vector. These two 
representations are connected by the automorphism 
$G_r \to - G_r$ and become equivalent when restricted to the even 
(Bose) subalgebra. Accordingly $(c_N,h_{p,q})_{R}$ denotes indifferently these 
two representations which are clearly inequivalent when 
$h_{p,q} \neq \frac{c_N}{24}$. 

For a given $N$ the equality $h_{p,q} = h_{p',q'}$ when 
$p-q$ and $p'-q'$ are both even or odd hold if and only 
if $p'=N+2-p$ and $q'= N+4 - q$. Note also that it may happen 
that $h_{p,q} = h_{p',q'}$ when $p-q$ is even and $p'-q$ 
is odd. For example, if $N=2$ then $h_{2,2} = h_{1,2}= 1/16$. 
Accordingly there are values of $N$ for which a given value 
of the lowest energy corresponds to three distinct irreducible 
representations of super-Virasoro algebra: one NS representation and two
R representations.  

From \cite[Section 4]{GKO} we can conclude that there exist DHR representations 
$\pi_{h_{p,q}}^{NS}$, $p-q$ even, and general solitons 
$\pi_{h_{p,q}}^{R}$, $p-q$ odd, of ${\rm SVir}_{c_N}$ (associated to the representations 
$(c_N,h_{p,q})_{NS}$, resp. $(c_N,h_{p,q})_{R}$ of the of super-Virasoro algebra)
such that  
\begin{equation}
\label{GKOnetsNS}
\left( \pi_{NS} \otimes \pi_{(N,\frac{1}{2}[p-1])}\right)
|_{\A_{\su2_{N+2}}\otimes {\rm SVir}_{c_N}} =
\bigoplus_q \pi_{(N+2, \frac{1}{2}[q-1])} \otimes \pi_{h_{p,q}}^{NS},
\end{equation}
$1\leq q \leq N+3$, $p-q$ even, and
\begin{equation}
\label{GKOnetsR}
\left( \pi_{R} \otimes \pi_{(N,\frac{1}{2}[p-1])}\right)
|_{\A_{\su2_{N+2}}\otimes {\rm SVir}_{c_N}} =
\bigoplus_q \pi_{(N+2, \frac{1}{2}[q-1])} \otimes \pi_{h_{p,q}}^{R},
\end{equation}
$1\leq q \leq N+3$, $p-q$ odd. 
                                                                                
We now denote $\rho_{h_{p,q}}^{NS}$,  resp. $\rho_{h_{p,q}}^{R}$, 
the restriction of $\pi_{h_{p,q}}^{NS}$, resp 
$\pi_{h_{p,q}}^{R}$, to ${\rm SVir}^0_{c_N}$. 

In the representation space of $\pi_{h_{p,q}}^{NS}$ the grading is always 
unitarily implemented and hence we have the direct sum 
$$\rho_{h_{p,q}}^{NS}= \rho_{h_{p,q}}^{NS+} \oplus\rho_{h_{p,q}}^{NS-}$$ 
of two (inequivalent) irreducible representations corresponding to the 
eigenspaces with eigenvalues 1 and -1 of the grading operator 
respectively. 

In contrast in the case of $\pi_{h_{p,q}}^{R}$ the grading automorphism 
is unitarily implemented only if $h_{p,q}=c_N/24$. This happens  
if and only if $N$ is even and $p=(N+2)/2$, $q=(N+4)/2$. 
In this case $\pi_{\frac{c_N}{24}}^R$ is a supersymmetric 
general representation of the Fermi conformal net ${\rm SVir}_{c_N}$.
Moreover we have the decomposition into irreducible (inequivalent) 
subrepresentations $$\rho_{\frac{c_N}{24}}^R= \rho_{\frac{c_N}{24}}^{R+}\oplus 
\rho_{\frac{c_N}{24}}^{R-}.$$ In the remaining cases 
$\rho_{h_{p,q}}^{R}$ is irreducible. 

Restricting Eq. (\ref{GKOnetsNS} ) and Eq. (\ref{GKOnetsR} ) to the Bose 
elements and using Equations (\ref{NSrest}), (\ref{Rrest})  we get 
\begin{equation}
\label{GKOnetsNSb+}
\left( \pi_{(2,0)} \otimes \pi_{(N,\frac{1}{2}[p-1])}\right)
|_{\A_{\su2_{N+2}}\otimes {\rm SVir}^0_{c_N}
} =
\bigoplus_q \pi_{(N+2, \frac{1}{2}[q-1])} \otimes \rho_{h_{p,q}}^{NS+},
\end{equation}
\begin{equation}
\label{GKOnetsNSb-}
\left( \pi_{(2,1)} \otimes \pi_{(N,\frac{1}{2}[p-1])}\right)
|_{\A_{\su2_{N+2}}\otimes {\rm SVir}^0_{c_N}} =
\bigoplus_q \pi_{(N+2, \frac{1}{2}[q-1])} \otimes \rho_{h_{p,q}}^{NS-},
\end{equation}
$1\leq q \leq N+3$, $p-q$ even, and 
\begin{equation}
\label{GKOnetsRb}
\left( \pi_{(2,\frac{1}{2})} \otimes \pi_{(N,\frac{1}{2}[p-1])}\right)
|_{\A_{\su2_{N+2}}\otimes {\rm SVir}^0_{c_N}} =
\bigoplus_q \pi_{(N+2, \frac{1}{2}[q-1])} \otimes \rho_{h_{p,q}}^{R},
\end{equation}
$1\leq q \leq N+3$, $p-q$ odd.     

Now, recalling the identification of ${\rm SVir}^0_{c_N}$ as a 
coset in Eq. (\ref{BoseCoset} ), it 
follows from \cite[Corollary 3.2]{X2} that every irreducible 
DHR representation of this net is equivalent to one of those considered 
before, namely $\rho_{h_{p,q}}^{NS+}$, $\rho_{h_{p,q}}^{NS-}$
and $\rho_{h_{p,q}}^{R}$ ($h_{p,q}\neq c_N/24$), $\rho_{\frac{c_N}{24}}^{R+}$
and $\rho_{\frac{c_N}{24}}^{R-}$.  
\subsection{Modularity of local super-Virasoro nets}
We state here explicitly the modularity of the Bose subnet super-Virasoro nets for $c<3/2$. 
In this case the Bose super-Virasoro net can be obtained as the coset \eqref{BoseCoset}. 
Then the Rehren $S$ and $T$ matrices as been computed by Xu in \cite[Sect.2.2.]{X3} 
(see also Sect. \ref{classification} below). These matrices agree with those in \cite{GW} and 
\cite{FSS} giving modular transformations of specialised characters. Accordingly we have 
the following. 
\begin{theorem} 
For a positive even integer $N$ then ${\rm SVir}^0_{c_N}$ is a modular conformal net. 
\end{theorem}
\section{Classification of superconformal nets in the discrete series}
\label{classification}
By a \emph{superconformal net} (of von Neumann algebras on $S^1$) 
we shall mean a Fermi net on $S^1$ that contains
a super-Virasoro net as irreducible subnet.
If the central charge $c$ of a superconformal net
is less than $3/2$, it is of the form
$c=\displaystyle\frac{3}{2}\left(1-\frac{8}{m(m+2)}\right)$
for some $m=3,4,5,\dots$ \cite{FQS}.  We classify all such superconformal
nets.

\subsection{Outline of classification}
As above, we denote the super Virasoro net with central charge $c$
and its Bosonic part by $\SVir_c$ and $\SVir^0_c$, respectively.
We are interested in the case $c<3/2$.  In this case, we have
$c=\displaystyle\frac{3}{2}\left(1-\frac{8}{m(m+2)}\right)$
for some $m=3,4,5,\dots$, and we have already seen that in this case
the local conformal net $\SVir^0_c$ is realised as a coset
net for the inclusion $SU(2)_m\subset SU(2)_{m-2}\otimes SU(2)_2$.
This net is completely rational in the sense of \cite{KLM}
by \cite{X2}.  The DHR sectors of the local conformal net $\SVir^0_c$ is
described as follows by \cite[Section 2.2]{X3}.
Label the DHR sectors of the local conformal
nets $SU(2)_m$, $SU(2)_{m-2}$, and $SU(2)_2$ by $k=0,1,\dots, m$,
$j=0,1,\dots,m-2$, and $l=0,1,2$, respectively.  Then we consider the
triples $(j,k,l)$ with $j-k+l$ being even.
For $l=0,2$, we have identification
$$(j,k,l)\leftrightarrow(m-2-j,m-k,2-l),$$
thus it is enough to consider the triples $(j,k,0)$ with $j-k$
being even.  Each such triple labels an irreducible DHR sector
of the coset net $\SVir^0_c$.  For the case $l=1$, we also have
identification
$$(j,k,l)\leftrightarrow(m-2-j,m-k,2-l),$$
but if we have a fixed point for this symmetry, that is,
if $m$ is even, then the fixed point $(2/m-1, m/2,1)$ splits
into two pieces, $(2/m-1, m/2,1)_+$ and $(2/m-1, m/2,1)_-$.
All of these triples, with this identification
and splitting, label all the irreducible DHR sectors of
the coset net $\SVir^0_c$.  The sectors with $l=0$ and $l=1$
are called Neveu-Schwarz and Ramond sectors, respectively.

The conformal spin of the sector $(j,k,l)$ is given by
$$\exp\left(\frac{\pi i}{2}\left(\frac{j(j+2)}{m}-
\frac{k(k+2)}{m+2}+\frac{l(l+2)}{4}\right)\right).$$
(This also works for the case $(j,k,l)=(2/m-1, m/2,1)$.) 

For example, if $m=3$, we have six
irreducible DHR sectors and they are labelled with triples
$(0,0,0)$, $(0,2,0)$, $(1,1,0)$, $(1,3,0)$, $(0,3,1)$, $(1,2,1)$.
(This local conformal net is equal to the Virasoro net with
$c=7/10$.)
For $m=4$, we have 13 irreducible DHR sectors and they are
labelled with
$(0,0,0)$, $(0,2,0)$, $(0,4,0)$, $(1,1,0)$, $(1,3,0)$,
$(2,0,0)$, $(2,2,0)$, $(2,4,0)$, $(0,3,1)$, $(1,4,1)$,
$(2,3,1)$, $(1,2,1)_+$, $(1,2,1)_-$, where
the two labels $(1,2,1)_+$, $(1,2,1)_-$ arise from the fixed
point $(1,2,1)$ of the symmetry of order 2.

For all $m$, the irreducible DHR sector $(m-2,m,0)$ has a
dimension 1 and a spin $-1$.  The superconformal net
$\SVir_c$ arises as a non-local extension of $\SVir^0_c$
as a crossed product by $\Z_2$ using identity and this sector.

Let $\A$ be any superconformal net on the circle with $c<3/2$.
Then let $\B$ be its Bosonic part.  By a similar argument to that
in \cite[Proposition 3.5]{KL1}, we know that the local conformal
net $\B$ is an irreducible extension of the local conformal
net $\SVir^0_c$, where $c$ is the central charge of $\A$.
By the strategy in \cite{KL1} based on \cite{BEK1},
we know that the dual canonical endomorphism $\th$ of an extension
is of the form $\th=\sum_{\la}Z_{0,\la}\la$, where $Z$ is the
modular invariant arising from the extension as in \cite{BEK1}.
Cappelli \cite{Ca} gave a list of type I modular invariants and
conjectured that it is a complete list.  From his list,
it is easy to guess that the dual canonical endomorphisms
we use for obtaining extensions are those listed
in Table \ref{dual-can}.
(Cappelli also considered type II modular invariants, but they
do not correspond to local extensions, so we ignore them here.)
\begin{table}[htbp]
\begin{center}
\begin{tabular}{|c|l|l|c|}\hline
& $m$ & $\th$ & Label
\\ \hline
(1) & any $m$ & $\th=(0,0,0)$ & $(A_{m-1},A_{m+1})$ \\ \hline
(2) & $m=4m'$ & $\th=(0,0,0)\oplus(0,m,0)$ & $(A_{4m'-1}, D_{2m'+2})$ 
\\ \hline
(3) & $m=4m'+2$ & $\th=(0,0,0)\oplus(m-2,0,0)$ & $(D_{2m'+2}, A_{4m'+3})$ 
\\ \hline
(4) & $m=10$ & $\th=(0,0,0)\oplus(0,6,0)$ & $(A_9, E_6)$ \\ \hline
(5) & $m=12$ & $\th=(0,0,0)\oplus(6,0,0)$ & $(E_6, A_{13})$ \\ \hline
(6) & $m=28$ & $\th=(0,0,0)\oplus(0,10,0)\oplus(0,18,0)\oplus(0,28,0)$
& $(A_{27},E_8)$ \\ \hline
(7) & $m=30$ & $\th=(0,0,0)\oplus(10,0,0)\oplus(18,0,0)\oplus(28,0,0)$
& $(E_8, A_{31})$ \\ \hline
(8) & $m=10$ & $\th=(0,0,0)\oplus(0,6,0)\oplus(8,6,0)\oplus(8,6,0)$
& $(D_6, E_6)$ \\ \hline
(9) & $m=12$ & $\th=(0,0,0)\oplus(6,0,0)\oplus(0,12,0)\oplus(6,12,0)$
& $(E_6, D_8)$ \\ \hline
\end{tabular}
\caption{List of candidates of the dual canonical endomorphisms}
\label{dual-can}
\end{center}
\end{table}
We will prove that each of the dual canonical endomorphisms 
in Table \ref{dual-can} gives a local extension of $\SVir^0_c$ 
in a unique way and that an arbitrary such local extension of
$\SVir^0_c$ gives one of the dual canonical endomorphisms 
in Table \ref{dual-can}.

\subsection{Study of type I modular invariants} 
We study type I modular invariants for the coset
nets for the inclusions $SU(2)_m\subset SU(2)_{m-2}\otimes SU(2)_2$.

First we recall the $S$ and $T$ matrices for $SU(2)_m$.
For $j,k=0,1,2,\dots,m$, we have the following.
\begin{eqnarray*}
S^{(m)}_{jk}&=&\sqrt{\frac{2}{m+2}}\sin \pi\frac{(j+1)(k+1)}{m+2},\\
T^{(m)}_{jk}&=&\de_{jk}\exp\frac{\pi i}{2}
\left(\frac{(j+1)^2}{m+2}-\frac12\right).
\end{eqnarray*}
For odd $m$, we have no problem arising from a fixed point of
the order two symmetry, and in this case, the modular invariants
have been already classified by Gannon-Walton \cite{GW}, which
shows that the identity matrix is the only modular invariant.  So
we have no non-trivial extensions in these cases.

So we now deal with the case of even $m$ in the rest of this
section and put $m=2m_0$.  In this case, the
$S$-matrix of the modular tensor category of the irreducible
DHR-sectors of the coset net is already not so easy to
obtain, and it has been computed by Xu \cite{X3}.

As in the previous section, we label the irreducible DHR
sectors of the coset net for the inclusion
$SU(2)_m\subset SU(2)_{m-2}\otimes SU(2)_2$ with
triples $(j,k,l)$ with $j-k+l$ being even, except for the
case $(m/2-1, m/2, 1)$.  For this fixed point of the order
two symmetry, we use labels $\phi_1=(m/2-1, m/2, 1)_+$, 
$\phi_2=(m/2-1, m/2, 1)_-$ to denote the two
irreducible DHR sectors.  We also use the symbols
$S^{(m-2)}$, $S^{(m)}$, $S^{(2)}$ and $S$ for the $S$-matrices
for the nets $SU(2)_{m-2}$, $SU(2)_m$, $SU(2)_2$ and the coset
for $SU(2)_m\subset SU(2)_{m-2}\otimes SU(2)_2$, respectively.
Then Xu's computation for the matrix
$S$ in \cite{X3} gives the following, where $n=1,2$.
\begin{enumerate}
\item $S_{(j,k,l),(j',k',l')}=2S^{(m-2)}_{jj'} \overline{S^{(m)}_{kk'}}
S^{(2)}_{ll'}$ for $(j,k,l), (j',k',l')\neq \phi_{1,2}$.
\item $S_{(j,k,l),(j',k',l'),\phi_n}=S_{\phi_n,(j,k,l),(j',k',l')}=
S^{(m-2)}_{j,m/2-1} \overline{S^{(m)}_{k,m/2}}
S^{(2)}_{l1}$ for $(j,k,l)\neq \phi_{1,2}$.
\item $S_{\phi_n,\phi_{n'}}=\delta_{nn'}+(S^{(m-2)}_{m/2-1,m/2-1}
\overline{S^{(m)}_{m/2,m/2}}S^{(2)}_{11}-1)/2$.
\end{enumerate}
The $T$-matrix of the coset is described as follows more easily.
\begin{enumerate}
\item $T_{(j,k,l),(j',k',l')}=T^{(m-2)}_{jj'} \overline{T^{(m)}_{kk'}}
T^{(2)}_{ll'}$ for $(j,k,l), (j',k',l')\neq \phi_{1,2}$.
\item $T_{(j,k,l),\phi_n}=T_{\phi_n,(j,k,l)}=
T^{(m-2)}_{j,m/2-1} \overline{T^{(m)}_{k,m/2}}
T^{(2)}_{l1}$ for $(j,k,l)\neq \phi_{1,2}$.
\item $T_{\phi_n,\phi_n'}=\delta_{nn'}
T^{(m-2)}_{m/2-1,m/2-1} \overline{T^{(m)}_{m/2,m/2}}
T^{(2)}_{11}$.
\end{enumerate}
Suppose we have a modular invariant $Z$ for the coset
for $SU(2)_m\subset SU(2)_{m-2}\otimes SU(2)_2$.
We define a new matrix  $\tilde Z_{(j,k,l),(j',k',l')}$ where
the triples $(j,k,l)$ and $(j',k',l')$ satisfy
$j,j'\in \{0,1,\dots,m-2\}$,
$k,k'\in \{0,1,\dots,m\}$,
$l,l'\in \{0,1,2\}$.  Note that we have no identification
or splitting for the triples $(j,k,l)$ and $(j',k',l')$ here.
\begin{enumerate}
\item If $j-k+l, j'-k'+l'\in 2\Z$ and
$(j,k,l), (j',k',l')\neq(m/2-1,m/2,1)$, then
we set $\tilde Z_{(j,k,l),(j',k',l')}=Z_{(j,k,l),(j',k',l')}$.
\item If $j-k+l\in 2\Z$ and
$(j,k,l)\neq(m/2-1,m/2,1)$, then
we set $\tilde Z_{(j,k,l),(m/2-1,m/2,1)}=
Z_{(j,k,l),\phi_1}+Z_{(j,k,l),\phi_2}$.
\item If $j'-k'+l'\in 2\Z$ and
$(j',k',l')\neq(m/2-1,m/2,1)$, then
we set $\tilde Z_{(m/2-1,m/2,1),j'k'l'}=
Z_{\phi_1,(j',k',l')}+Z_{\phi_2,j'k'l'}$.
\item $\tilde Z_{(m/2-1,m/2,1),(m/2-1,m/2,1)}=
Z_{\phi_1,\phi_1}+Z_{\phi_1,\phi_2}+Z_{\phi_2,\phi_1}+Z_{\phi_2,\phi_2}$.
\item If $j-k+l\notin 2\Z$ or $j'-k'+l'\notin 2\Z$, then
we set $\tilde Z_{(j,k,l),(j',k',l')}=0$.
\end{enumerate}
This construction is analogous to the one of ${\cal L}^{cw}$
in \cite[page 178]{GW}, but now
unlike in \cite{GW}, our map $Z\mapsto \tilde Z$ may not be 
injective because of the definition of $\tilde Z$ involving
the row/column labelled with $(m/2-1,m/2,1)$.

For triples $(j,k,l)$ and $(j',k',l')$ satisfying
$j,j'\in \{0,1,\dots,m-2\}$,
$k,k'\in \{0,1,\dots,m\}$,
$l,l'\in \{0,1,2\}$, we set as follows.
(We do not impose the conditions $j-k+l, j'-k'+l'\in 2\Z$ here.)
\begin{enumerate}
\item $\tilde S_{(j,k,l),(j',k',l')}=S^{(m-2)}_{jj'}
\overline{S^{(m)}_{kk'}} S^{(2)}_{ll'}$.
\item $\tilde T_{(j,k,l),(j',k',l')}=T^{(m-2)}_{jj'}
\overline{T^{(m)}_{kk'}} T^{(2)}_{ll'}$.
\end{enumerate}
We also write $\phi=(m/2-1,m/2,1)$ and set
$$I=\{(j,k,0)\mid j+k\in2\Z\}\cup
\{(j,k,1)\mid j+k\notin2\Z, j+k<m-1\},$$
which gives representatives of
$\{(j,k,l)\mid j-k+l\in2\Z\}\setminus\{\phi\}$ under the
order 2 symmetry.

We now have the following proposition.
\begin{proposition}
If the matrix $Z$ is a modular invariant for $S, T$, then
the matrix $\tilde Z$ is a modular invariant for $\tilde S, \tilde T$.
\end{proposition}
\begin{proof}
It is clear that the entries of the matrix $\tilde Z$ are
nonnegative integers and we have $\tilde Z_{(0,0,0), (0,0,0)}=1$.
We thus have to prove $\tilde Z \tilde S=\tilde S \tilde Z$ and
$\tilde Z \tilde T=\tilde T \tilde Z$.
The latter is clear from the definition, so we have to prove
the former identity at $(j,k,l)$, $(j',k',l')$.
We deal with seven cases separately.

Case (1).  $j-k+l\notin2\Z$, $j'-k'+l'\notin2\Z$.
The identity is trivial because the both hand sides are zero.

Case (2).  $j-k+l\in2\Z$, $j'-k'+l'\notin2\Z$.
The right hand side is obviously zero by the definition of $\tilde Z$,
and the following three identities now prove that the left hand
side is also zero.
\begin{eqnarray*}
\tilde S_{\phi,(j',k',l')}&=&0,\\
\tilde Z_{(j,k,l),(j'',k'',l'')}&=&
\tilde Z_{(j,k,l),(m-2-j'',m-k'',2-l'')},\\
\tilde S_{(j'',k'',l''),(j',k',l')}&=&
-\tilde S_{(m-2-j'',m-k'',2-l''),(j',k',l')}.
\end{eqnarray*}
Note that the third identity holds because $j'-k'+l'\notin2\Z$.

Case (3).  $j-k+l\notin2\Z$, $j'-k'+l'\in2\Z$.
This case can be proved in the same way as above.

Case (4).  $j-k+l\in2\Z$, $j'-k'+l'\in2\Z$,
$(j,k,l)\neq \phi$, $(j',k',l')\neq \phi$.
The identity $ZS=SZ$ gives the following identity.
\begin{eqnarray*}
&&\sum_{(j'',k'',l'')\in I} 2Z_{(j,k,l),(j'',k'',l'')}
\tilde S_{(j'',k'',l''),(j',k',l')}+
Z_{(j,k,l),\phi_1}\tilde S_{\phi_1,(j',k',l')}+
Z_{(j,k,l),\phi_2}\tilde S_{\phi_2,(j',k',l')}\\
&=&\sum_{(j'',k'',l'')\in I} 2S_{(j,k,l),(j'',k'',l'')}
Z_{(j'',k'',l''),(j',k',l')}+
\tilde S_{(j,k,l),\phi} Z_{\phi_1,(j',k',l')}+
\tilde S_{(j,k,l),\phi} Z_{\phi_2,(j',k',l')}.
\end{eqnarray*}
Since we now have
\begin{eqnarray*}
\tilde S_{(j'',k'',l''),(j',k',l')}&=&
\tilde S_{(m-2-j'',m-k'',2-l''),(j',k',l')},\\
\tilde S_{(j,k,l),(j'',k'',l'')}&=&
\tilde S_{(j,k,l),(m-2-j'',m-k'',2-l'')}
\end{eqnarray*}
we conclude $(\tilde Z\tilde S)_{(j,k,l),(j',k',l')}
=(\tilde Z\tilde Z)_{(j,k,l),(j',k',l')}$.

Case (5).  $j-k+l\in2\Z$, $(j,k,l)\neq \phi$, $(j',k',l')=\phi$.
The identity $ZS=SZ$ gives the following identity for $n=1,2$.
\begin{eqnarray*}
&&\sum_{(j'',k'',l'')\in I} Z_{(j,k,l),(j'',k'',l'')}
\tilde S_{(j'',k'',l''),\phi}+
Z_{(j,k,l),\phi_1} S_{\phi_1,\phi_n}+
Z_{(j,k,l),\phi_2} S_{\phi_2,\phi_n}\\
&=&\sum_{(j'',k'',l'')\in I} 2\tilde S_{(j,k,l),(j'',k'',l'')}
Z_{(j'',k'',l''),\phi_n}+
\tilde S_{(j,k,l),\phi} Z_{\phi_1,\phi_n}+
\tilde S_{(j,k,l),\phi} Z_{\phi_2,\phi_n}.
\end{eqnarray*}

Adding these two equalities for $n=1,2$, we obtain the following
identity, which gives $(\tilde Z\tilde S)_{(j,k,l),\phi}
=(\tilde Z\tilde Z)_{(j,k,l),\phi}$.
\begin{eqnarray*}
&&\sum_{(j'',k'',l'')\in I} Z_{(j,k,l),(j'',k'',l'')}
\tilde S_{(j'',k'',l''),\phi}+
(Z_{(j,k,l),\phi_1}+Z_{(j,k,l),\phi_2}) \tilde S_{\phi,\phi}\\
&=&\sum_{(j'',k'',l'')\in I} 2\tilde S_{(j,k,l),(j'',k'',l'')}
(Z_{(j'',k'',l''),\phi_1}+Z_{(j'',k'',l''),\phi_2})+\\
&&\quad\quad\quad\quad\tilde S_{(j,k,l),\phi}
(Z_{\phi_1,\phi_1}+Z_{\phi_1,\phi_2}+Z_{\phi_2,\phi_1}+Z_{\phi_2,\phi_2}).
\end{eqnarray*}

Case (6).  $(j,k,l)=\phi$, $j'-k'+l'\in2\Z$, $(j',k',l')\neq \phi$.
This case can be proved in the same way as above.

Case (7).  $(j,k,l)=(j',k',l')=\phi$.
The identity $ZS=SZ$ gives the following identity for $n=1,2$, $n'=1,2$.
\begin{eqnarray*}
&&\sum_{(j'',k'',l'')\in I} Z_{\phi_n,(j'',k'',l'')}
\tilde S_{(j'',k'',l''),\phi}+
Z_{\phi_n,\phi_1} S_{\phi_1,\phi_{n'}}+
Z_{\phi_n,\phi_2} S_{\phi_2,\phi_{n'}}\\
&=&\sum_{(j'',k'',l'')\in I} \tilde S_{\phi,(j'',k'',l'')}
Z_{(j'',k'',l''),\phi_{n'}}+
S_{\phi_n,\phi_1} Z_{\phi_1,\phi_{n'}}+
S_{\phi_n,\phi_2} Z_{\phi_2,\phi_{n'}}.
\end{eqnarray*}
Adding these four identities for $n,n'=1,2$, we obtain the following
identity, which gives $(\tilde Z\tilde S)_{\phi,\phi}
=(\tilde Z\tilde Z)_{\phi,\phi}$.
\begin{eqnarray*}
&&\sum_{(j'',k'',l'')\in I} 
2(Z_{\phi_1,(j'',k'',l'')}+Z_{\phi_2,(j'',k'',l'')})
\tilde S_{(j'',k'',l''),\phi}+\\
&&\quad\quad\quad\quad
(Z_{\phi_1,\phi_1}+Z_{\phi_1,\phi_2}+Z_{\phi_2,\phi_1}+Z_{\phi_2,\phi_2})
\tilde S_{\phi,\phi}\\
&=&\sum_{(j'',k'',l'')\in I} 2\tilde S_{\phi,(j'',k'',l'')}
(Z_{(j'',k'',l''),\phi_1}+Z_{(j'',k'',l''),\phi_2})+\\
&&\quad\quad\quad\quad\tilde S_{\phi,\phi}
(Z_{\phi_1,\phi_1}+Z_{\phi_1,\phi_2}+Z_{\phi_2,\phi_1}+Z_{\phi_2,\phi_2})
\end{eqnarray*}
\end{proof}
We now recall Gannon's parity rule \cite[page 696]{G}.
For $n$ and $j$, we set $\e_{2n}(j)$ as follows.
\begin{enumerate}
\item If $j\equiv 0,n$ mod $2n$, then $\e_{2n}(j)=0$.
\item If $j\equiv 1,2,\dots,n-1$ mod $2n$, then $\e_{2n}(j)=1$.
\item If $j\equiv n+1,n+2,\dots,2n-1$ mod $2n$, then $\e_{2n}(j)=-1$.
\end{enumerate}
Set $$L=\{n\mid (n,8m(m+2))=1\}.$$
Then the parity rule is the following.
If we have a modular invariant $\tilde Z$ for $\tilde S$ and $\tilde T$
with $tilde Z_{000,jkl}\neq0$, then
the parity rule says that we have
$$\e_{2m}(n)\e_{2m+4}(n)\e_{8}(n)=
\e_{2m}(n(j+1))\e_{2m+4}(n(k+1))\e_{8}(n(l+1))$$
for all $n\in L$.

Now suppose $m_0=m/2$ is odd.
For any $n_1$ with $(n_1,4m_0)=1$, choose $p\in\{1,-1\}$ so
that $n_1 \equiv p$ mod $4$.  Then it is easy to see that
there exists $n\in L$ satisfying
$n\equiv n_1$ mod $4m_0$,
$n\equiv p$ mod $4(m_0+1)$,
$n\equiv p$ mod $8$, since we have $(m_0, m_0+1)=1$.
Suppose $tilde Z_{000,jkl}\neq0$.  Then the parity rule says
$$\e_{2m}(n_1)\e_{2m+4}(p)\e_{8}(p)=
\e_{2m}(n_1(j+1))\e_{2m+4}(p(k+1))\e_{8}(p(l+1)),$$
which gives $\e_{2m}(n_1)=\e_{2m}(n_1(j+1))$
for all $n_1$ with $(n_1,2m)=1$.  If $m_0\neq 3,5,15$, then
we have $j=0,m-2$ by \cite{G}.  Using the order two symmetry
$(j,k,l)\leftrightarrow(m-2-j,m-k,2-l)$ and the
definition of $\tilde Z$,
we conclude that our dual canonical endomorphism $\th$ has a 
decomposition within $\{(0,k,l)\mid -k+l\in2\Z\}.$
Let ${\cal T}_1$ be the tensor category of the representations of
$SU(2)_m$ with the opposite braiding and ${\cal T}_2$ be the
tensor category of the representations of $SU(2)_2$ with the
usual braiding.  If $\th$ gives a local extension for the
tensor category generated by $\{(0,k,l)\mid -k+l\in2\Z\}$, then
it also gives a local extension for the tensor category
${\cal T}_1 \times {\cal T}_2$ by the description of the
braided tensor category having
the simple objects $\{(0,k,l)\mid -k+l\in2\Z\}$ by
\cite[Section 4.3]{X1},
\cite[Proposition 2.3.1, Proof of Theorem B]{X5}.
(Note that we have the notion
of a local extension as a local $Q$-system in 
\cite{LR1} for an abstract tensor category.)  Then this extension
for ${\cal T}_1 \times {\cal T}_2$ gives a type I modular invariant
for the $S$- and $T$-matrices arising from ${\cal T}_1 \times {\cal T}_2$.
Then Gannon's classification of modular invariants for
$SU(2)_m\otimes SU(2)_2$ in \cite[Section 7]{G} shows that
our $\th$ must be one of Table \ref{dual-can}.
(Here we have an opposite braiding for ${\cal T}_1$, but this
does not matter since we can consider a modular invariant
$Z_{k'l,kl'}$ instead of $Z_{kl,k'l'}$ as in \cite[page 178]{GW}.)

Next suppose $m_0=m/2$ is even.
For any $n_2$ with $(n_2,4(m_0+1))=1$, choose $p\in\{1,-1\}$ so
that $n_2 \equiv p$ mod $4$.  Then it is easy to see that
there exists $n\in L$ satisfying
$n\equiv p$ mod $4m_0$,
$n\equiv n_2$ mod $4(m_0+1)$,
$n\equiv p$ mod $8$, since we have $(m_0, m_0+1)=1$.
Suppose $\tilde Z_{(0,0,0),(j,k,l)}\neq0$.  Then the parity rule says
$$\e_{2m}(p)\e_{2m+4}(n_2)\e_{8}(p)=
\e_{2m}(p(j+1))\e_{2m+4}(n_2(k+1))\e_{8}(p(l+1)),$$
which gives $\e_{2m+4}(n_2)=\e_{2m+4}(n_2(k+1))$
for all $n_2$ with $(n_2,4(m_0+1))=1$.  If $m_0\neq 2,4,14$, then
we have $k=0,m$ by \cite{G}.  Using the order two symmetry
$(j,k,l)\leftrightarrow(m-2-j,m-k,2-l)$ and the
definition of $\tilde Z$ again,
we conclude that our dual canonical endomorphism $\th$ has a 
decomposition within $\{(j,0,l)\mid j+l\in2\Z\}$.
If $\th$ gives a local extension for the
tensor category generated by $\{(0,k,l)\mid -k+l\in2\Z\}$, then
it also gives a local extension for the tensor category
of the representations $SU(2)_{m-2}\otimes SU(2)_2$ by
the description of the braided tensor category having
the simple objects $\{(j,0,l)\mid j+l\in2\Z\}$ by
\cite[Section 4.3]{X1},
\cite[Proposition 2.3.1]{X6}.  Then this extension
gives a type I modular invariant
for the $S$- and $T$-matrices arising from $SU(2)_{m-2}\otimes SU(2)_2$,
so Gannon's classification of modular invariants for
$SU(2)_{m-2}\otimes SU(2)_2$ in \cite[Section 7]{G} again shows that
our $\th$ must be one of Table \ref{dual-can}.

We now deal with the remaining exceptional cases $m=4,6,8,10,28,30$.
If $\tilde Z_{(0,0,0),(j,k,l)}\neq0$ and $j-k+l\in2\Z$, then the
conformal spin of $(j,k.l)$ is 1.  This condition already
gives the restriction that the dual canonical endomorphisms
$\th$ has a decomposition within
$$\{(j,k,0)\mid j+l\in2\Z\}$$
for $m=4,6,8,10,28$.  Again Gannon's classification of modular invariants
in \cite{G} shows that our $\th$ must be one of Table \ref{dual-can}.

The final remaining case is $m=30$.  In this case, 
the only sectors $(j,k,l)$ with $l=1$ and $j-k+l\in2\Z$
having conformal spins 1 are
$(3,12,1)$ and $(13,18,1)$.
The parity rules with $n=13, 7$
exclude these two sectors from $\th$, respectively.  Thus again
the dual canonical endomorphism $\th$ has a decomposition within
$$\{(j,k,0)\mid j+l\in2\Z\}$$
and Gannon's classification in \cite{G} shows that our $\th$ must
be one of Table \ref{dual-can}.

We have thus proved that if we have a local extension of
$\SVir^0_c$, then its dual canonical endomorphism must be one of
those listed in Table \ref{dual-can}.

\subsection{Classification of the Fermi extensions of the 
super-Vi\-ra\-so\-ro net, $c<3/2$}

We now show that each endomorphism in Table \ref{dual-can}
uniquely produces a superconformal net extending $\SVir_c$.

(1) We have the label $(A_{m-1}, A_{m+1})$.
The identity matrix obviously gives the coset net for
$SU(2)_m\subset SU(2)_{m-2}\otimes SU(2)_2$ itself.  This has
a unique Fermionic extension.

(2) Now we have the label $(A_{4m'-1}, D_{2m'+2})$ with $m=4m'$.
The irreducible DHR sectors $(0,0,0)$, $(0,4m',0)$ give the
group $\Z_2$.  The crossed product of the coset net by this group
gives a local extension.  The irreducible DHR sector $(4m'-2,0,0)$
has spin $-1$, and this sector still gives an irreducible DHR sector
of index 1 and spin $-1$ after the $\a$-induction applied to
the crossed product construction.
This $\a$-induced sector gives a unique Fermionic extension
through a crossed product by  $\Z_2$.

(3) Now we have the label $(D_{2m'+2}, A_{4m'+3})$ with $m=4m'+2$.
The irreducible DHR sectors $(0,0,0)$, $(4m',0,0)$ give the
group $\Z_2$.  The crossed product of the coset net by this group
gives a local extension.  The irreducible DHR sector $(0,4m'+2,0)$
has spin $-1$, and this sector still gives an irreducible DHR sector
of index 1 and spin $-1$ after the $\a$-induction applied to
the crossed product construction.
This again gives a unique Fermionic extension in a similar way to
case (2).

(4) Now we have the label $(A_9, E_6)$ for $m=10$.
The DHR sector  $(0,0,0)\oplus(0,6,0)$ gives a local
$Q$-system as in \cite[Subsection 4.2]{KL1}
by the description of the braided tensor category having
the simple objects $\{(0,k,0)\mid k\in2\Z\}$ as in
\cite[Section 4.3]{X1},
\cite[Proposition 2.3.1, Proof of Theorem B]{X5}.

Uniqueness also follows as in
\cite[Remark 2.5]{KL1}, using the cohomology vanishing result
in \cite{KL2}.
(This is also a mirror extension in the sense of \cite{X6}, so
locality of the extension also follows from \cite{X6}.)
The irreducible DHR sector $(0,10,0)$
has spin $-1$, and this sector still gives an irreducible DHR sector
of index 1 and spin $-1$ after the $\a$-induction applied to
the local extension.  This gives a unique Fermionic extension
as above.

(5) Now we have the label $(E_6, A_{13})$ for $m=12$.
The irreducible DHR sectors  $(0,0,0)$, $(6,0,0)$ give a local
$Q$-system as in (4) by the description of the braided tensor
category having the simple objects $\{(j,0,0)\mid j\in2\Z\}$
as in \cite[Section 4.3]{X1}, \cite[Proposition 2.3.1]{X5}.
Uniqueness also follows as above.
(This is also a coset model as in \cite{KL1}.)
The irreducible DHR sector $(10,0,0)$
has spin $-1$, and this sector still gives an irreducible DHR sector
of index 1 and spin $-1$ after the $\a$-induction applied to
the local extension.  This gives a unique Fermionic extension
as above.

(6) Now we have the label $(A_{27}, E_8)$ for $m=28$.
The irreducible DHR sectors  $(0,0,0)$, $(0,10,0)$, $(0,18,0)$,
$(0,28,0)$ give a local
$Q$-system as in case (4).  Uniqueness also follows as above,
using the cohomology vanishing result in \cite{KL2}.
(This is again also a mirror extension in the sense of \cite{X6}.)
The irreducible DHR sector $(26,0,0)$
has spin $-1$, and this sector still gives an irreducible DHR sector
of index 1 and spin $-1$ after the $\a$-induction applied to
the local extension.  This gives a unique Fermionic extension
as above.

(7) Now we have the label $(E_8, A_{31})$ for $m=30$.
The irreducible DHR sectors  $(0,0,0)$, $(10,0,0)$, $(18,0,0)$,
$(28,0,0)$ give a local
$Q$-system as in (5).  Uniqueness also follows as above.
(This is again also a coset model as in \cite{KL1}.)
The irreducible DHR sector $(0,30,0)$
has spin $-1$, and this sector still gives an irreducible DHR sector
of index 1 and spin $-1$ after the $\a$-induction applied to
the local extension.  This gives a unique Fermionic extension
as above.

(8) Now we have the label $(D_6, E_6)$ for $m=10$.
The irreducible DHR sectors  $(0,0,0)$, $(0,6,0)$,
$(8,0,0)$, $(8,6,0)$ give a local
$Q$-system, since this is a further index 2 extension of the
local extension given in the above (4).
Uniqueness holds for this second extension.
The irreducible DHR sector $(0,10,0)$
has spin $-1$, and this sector still gives an irreducible DHR sector
of index 1 and spin $-1$ after the $\a$-induction applied to
the local extension.  This gives a unique Fermionic extension
as above.

(9) Now we have the label $(E_6, D_8)$ for $m=12$.
The irreducible DHR sectors  $(0,0,0)$, $(6,0,0)$,
$(0,12,0)$, $(6,12,0)$ give a local
$Q$-system, since this is a further index 2 extension of the
local extension given in the above (5).
Uniqueness holds for this second extension.
The irreducible DHR sector $(10,0,0)$
has spin $-1$, and this sector still gives an irreducible DHR sector
of index 1 and spin $-1$ after the $\a$-induction applied to
the local extension.  This gives a unique Fermionic extension
as above.

All the above give the following classification theorem.
\begin{theorem}\label{main}
The following gives a complete list of superconformal nets with
$c<3/2$ and the dual canonical endomorphisms for the Bosonic
part extending $\SVir^0_c$ in each case is given as in
Table \ref{dual-can}.
\begin{enumerate}
\item The super Virasoro net with
$c=\displaystyle\frac{3}{2}\left(1-\frac{8}{m(m+2)}\right)$,
labelled with $(A_{m-1}, A_{m+1})$.
\item Index 2 extensions of the above (1), labeled with
$(A_{4m'-1}, D_{2m'+2})$, $(D_{2m'+2}, A_{4m'+3})$.
\item Six exceptionals labelled with
$(A_9, E_6)$, $(E_6, A_{13})$, $(A_{27}, E_8)$, $(E_8, A_{31})$,
$(D_6, E_6)$, $(E_6, D_8)$.
\end{enumerate}
\end{theorem}
\appendix
\section{Appendixes}
\subsection{Topological covers and covering symmetries}
\label{top}
If $X$ is a connected manifold, we denote by DIFF$(X)$ the group of 
diffeomorphisms of $X$ and by $\Diff(X)$ the 
\emph{connected component of the identity} of DIFF$(X)$
\footnote{If $X$ is oriented, in general  $\Diff(X)$ is smaller than 
the of orientation preserving subgroup of DIFF$(X)$, see \cite{M}}.

Let $X$ be compact, connected manifold; then DIFF$(X)$ is an infinite dimensional Lie group modelled on the locally convex topological vector space Vect$(X)$.
Denote by $\tilde X$ the universal covering of $X$ and by $p:\tilde 
X\to X$ the covering map. The fundamental group $\Ga\equiv \pi_1(X)$ 
acts on $\tilde X$ by deck transformations: $p(\g x)=p(x)$ for 
$x\in\tilde X$ and $\g\in \Ga$.

Denote by $\Diff^{\Ga}(\tilde X)$ the subgroup of $\Diff(\tilde X)$ 
of diffeomorphisms commuting with the $\Ga$-action. Given 
$g\in\Diff^{\Ga}(\tilde X)$ there exists a unique $P(g)\in\Diff(X)$
\[
P(g)(x) \equiv p(g\tilde x), \ x\in X
\]
with $\tilde x$ any point in $\tilde X$ s.t. $p(\tilde x) = x$. The 
map $P: g\mapsto P(g)$ is clearly a continuous homomorphism of 
$\Diff^{\Ga}(\tilde X)$ into $\Diff(X)$. It is not difficult to check 
that $P$ is surjective. Indeed any $g\in {\rm DIFF}(X)$ has a lift to 
an element of ${\rm DIFF}^{\Ga}(\tilde X)$ by the uniqueness of $\tilde X$, 
thus $P$ extends to a surjective continuous homomorphism 
${\rm DIFF}^{\Ga}(\tilde X)\to {\rm DIFF}(X)$, that thus map the 
connected component of the identity to the 
connected component of the identity. More directly, if $g\in\Diff(X)$ 
is sufficiently small neighbourhood $\cal U$ the identity,
one can lift $g\in\cal U$ to an element of $\Diff^{\Ga}(\tilde X)$ because $p$ 
is a local diffeomorphism; as $\Diff(X)$ is generated by $\cal U$, 
all elements of $\Diff(X)$ have such a lift.

Denote by $\widetilde\Diff(X)$ the universal central cover of $\Diff(X)$. Then 
$\widetilde\Diff(X)$ is an infinite dimensional Lie group with center 
$Z\equiv \pi_1(\Diff(X))$ and we denote by $q$ the quotient map 
$q:\widetilde\Diff(X)\to \Diff(X)$ whose kernel is $Z$.

If $\{g_t\}_{t\in [0,1]}$ is a loop in $\Diff(X)$ with $g_0 = g_1 = 
\iota$ and $x_0\in X$, then $\{g_t(x_0)\}_{t\in [0,1]}$ 
is a loop in $X$ with $g_0(x_0) = g_1(x_0) = 
x_0$. Homotopic loops in $\Diff(X)$ clearly go to homotopic loops in 
$X$ and we thus obtain a homomorphism 
\begin{equation}\label{pi}
\Phi_0 :  \pi_1(\Diff(X)) \to  \pi_1(X)
\end{equation}
This homomorphism is certainly not surjective ($\pi_1(\Diff(X))$ is 
always abelian, $\pi_1(X)$ may be not). 
\begin{proposition} There exists a unique continuous homomorphism 
$\Phi:\widetilde\Diff(X)\to\Diff^{\Ga}(\tilde X)$ such that the following 
diagram commutes:
\[
\xymatrix{
\widetilde\Diff(X)  \ar[rr]^{\Phi} \ar[dr]^{q}& 
&  \Diff^{\Ga}(\tilde X) \ar[dl]_{P}  \\
& \Diff(X) }
\]
$\Phi$ restricts to a homomorphism
\begin{equation}\label{cen}
\Phi_0: \pi_1\big(\Diff(X)\big)\ (= Z) \to \text{\rm center of}\ \pi_1(X)
\end{equation}
that coincides with the one given above in eq. \eqref{pi}.
\end{proposition}
Thus $\pi_1(X)$ is a quotient of $\pi_1(\Diff(X))$. However $\Phi_0$ is 
not injective in general. For example $\pi_1(S^2)= \{\iota\}$ while 
$\pi_1\big(\Diff(S^2)\big)=\pi_1\big(\Diff(SO(3))\big)=\mathbb Z_2$ \cite{M} .

We end this appendix with the following proposition that has been used in the paper.
\begin{proposition}\label{locdiff}
The group $\Diff^{(2)}(S^1)$ is algebraically generated by $\Diff^{(2)}_I(S^1)$ as $I$ 
varies in $\I^{(2)}$.
\end{proposition}
\proof
As $I$ varies in $\I$, $\Diff_I(S^1)$ algebraically generates a normal subgroup of $\Diff(S^1)$, hence all $\Diff(S^1)$ because this is a simple group \cite{M}. With ${\rm rot}$ the rotation one-parameter subgroup of $\Diff(S^1)$, we denote by ${\rm rot}^{(2)}$ its lift to $\Diff^{(2)}(S^1)$.
Let $H$ be the subgroup of $\Diff^{(2)}(S^1)$ algebraically generated by the $\Diff^{(2)}_I(S^1)$'s and $z\equiv {\rm rot}^{(2)}(2\pi)$ the generator of the center of $\Diff^{(2)}(S^1)$ ($\simeq \mathbb Z_2$). With $q: \Diff^{(2)}(S^1)\to \Diff(S^1)$ the quotient map, we then have $q(H) = \Diff(S^1)$ and $\Diff^{(2)}(S^1)= H\cup Hz$. If $z\in H$ we then have $\Diff^{(2)}(S^1)= H$ as desired. 
So we may assume $z\notin H$. Then $q|_H$ is one-to-one, namely an isomorphism between $H$ and 
$\Diff(S^1)$.
Consider now the one-parameter subgroup of $H$ given by ${\rm rot}'(\th) \equiv q^{-1}({\rm rot(\th)})$. Clearly 
$q({\rm rot}' (\th) ) = {\rm rot} (\th) = q({\rm rot}^{(2)}  (\th) )$, so 
${\rm rot}^{(2)}(\th) = {\rm rot}'(\th) z^m$ where $m=0$ or $m=1$ depending on $\th$. So 
${\rm rot}^{(2)}(2\th) = {\rm rot}' (2\th)$ for all $\th$, thus  ${\rm rot}^{(2)}$ and ${\rm rot}'$ coincide, in particular $z={\rm rot}'(2\pi) = 1$ which is a contradiction.
\endproof
\subsection{McKean-Singer formula}
\label{MKS}
Let $\Gamma$ be a selfadjoint unitary on a Hilbert space $\H$, thus $\H = \H_+\oplus\H_-$ where $\H_{\pm}$ is the $\pm$-eigenspace of $\Ga$. A linear operator $Q$ on $\H$ is \emph{even} if it commutes with $\Ga$, namely $Q\H_\pm\subset\H_\pm$, and \emph{odd} if 
$\Gamma Q\Gamma^{-1} = - Q$, namely
\[
Q=\left[
\begin{matrix}
0 & Q_-\\
Q+- & 0
\end{matrix}
\right]
 \]
where $Q_{\pm}: \H_{\pm}\to\H_{\mp}$. Clearly $Q^*_+ = Q_-$ if $Q$ is 
selfadjoint. Note that $Q^2$ is even if $Q$ is odd.

We shall denote by $\Str = \Tr(\Ga\,\cdot)$ the \emph{supertrace} on $B(\H)$. Clearly 
the supertrace depends on $\Ga$, but the Hilbert spaces in the following will 
be equipped with a grading unitary on which the definition of $\Str$ will 
be naturally based. Note that if $T\in B(\H)$ then $\Str(T)$ is 
well-defined and finite iff $T$ is a trace class operator, i.e. $\Tr(|T|)<\infty$. Indeed $|\Ga T| = |T|$.

We recall the following well known lemma.
\begin{lemma} {\bf McKean-Singer formula.} Let $Q$ be a selfadjoint odd 
linear operator on $\H$. If $e^{-tQ^2}$ is trace class for some $t_0 >0$ 
then $\Str(e^{-tQ^2})$ is an integer independent of $t>t_0$ 
and indeed
\[
\Str(e^{-tQ^2}) = \ind(Q_+)
\]
where $\ind(Q_+)\equiv\Dim\ker(Q_+)-\Dim\ker(Q^{*}_+)$ is the Fredholm index of $Q_+$.
\end{lemma}
\proof
If $t>t_0$ both traces $\Tr(e^{-tQ^2})$ and $\Tr(Q^2 e^{-tQ^2})$ are 
convergent and we have
\[
\frac{\rm d}{{\rm d}t}\Str(e^{-tQ^2}) = -\Str(Q^2 e^{-tQ^2})\ .
\]
But $\Str(Q^2 e^{-tQ^2}) = 0$ because
\begin{multline*}
\Str(Q^2 e^{-tQ^2})
= \Tr(\Ga Q^2 e^{-tQ^2}) = - \Tr(Q\Ga Q e^{-tQ^2})\\
= - \Tr(\Ga Q e^{-tQ^2}Q) =- \Tr(\Ga Q^2 e^{-tQ^2})= -\Str(Q^2 e^{-tQ^2})\ ,
\end{multline*}
thus $\Str(e^{-tQ^2})$ is constant. Therefore
\[
\Str(e^{-tQ^2}) = \lim_{t\to +\infty}\Str(e^{-tQ^2}) = \Tr(\Ga E)
\]
where $E$ is the projection onto $\ker(Q^2)$. But $Q^2=Q^*_{+}Q_{+}\oplus Q^{*}_{-}Q_{-}$ in the graded decomposition of $\H$, thus $\ker(Q^2) = \ker(Q_+)\oplus\ker(Q_-)$ and so
\[
\Tr(\Ga E) = \Dim\ker(Q_{+}) - \Dim\ker(Q_{-}) = \ind(Q_+)\ .
\]
\endproof
\section{Outlook} A potential development of this work concerns the relation with the Noncommutative Geometrical framework of A. Connes \cite{C}. The supersymmetric sector with lowest weight $c/24$ in Section \ref{SVirnet} gives rise in a natural way to an infinite dimensional spectral triple whose Chern character is given by the JLO cyclic cocycle \cite{JLO}, but for a domain problem verification (see also \cite{BG} where an analogous domain problem is settled). According to the lines indicated in \cite{L4}, we expect the computation of this cocycle to be of interest in relation to a QFT index theorem. 
\bigskip

\smallskip

\noindent {{\small {\bf Acknowledgements.} 
The authors thank F. Xu for useful conversations. Our collaboration was carried on in different occasions:
R.L. wishes to thank V.G. Kac for the invitation at the Sch\"odinger Institute summer school in June 2005; S.C. and R.L.  thank Y. Kawahigashi for the invitations at the University of Tokyo in November 2005 and December 2006 respectively. } 

\end{document}